\documentclass[a4paper,UKenglish]{article}
\usepackage{hyperref}

\usepackage{amsmath,amsthm}
\usepackage{amssymb}
\usepackage[dvipsnames]{xcolor}
\usepackage{graphicx}
\usepackage{mathtools}
\usepackage{thm-restate}
\usepackage[margin=1in]{geometry}

\newtheorem{theorem}{Theorem}
\newtheorem{observation}[theorem]{Observation}
\newtheorem{corollary}[theorem]{Corollary}
\newtheorem{lemma}[theorem]{Lemma}
\newtheorem{claim}[theorem]{Claim}
\theoremstyle{definition}

\newtheorem{definition}[theorem]{Definition}

\newcommand{\B}[1]{\left( {#1} \right)}
\newcommand{\F}[2]{{\frac{#1}{#2}}}
\newcommand{\R}[1]{{\frac{1}{#1}}}
\newcommand{\BF}[2]{\B{\F{#1}{#2}}}

\newcommand{\set}[1]{\left\{ #1 \right\}}
\newcommand{\midline}[2]{{#1} \, \middle| \, {#2}}
\newcommand{\prob}[1]{\operatorname{Pr}\left( #1 \right)}

\newcommand{\stset}[2]{\set{\midline{#1}{#2}}}

\newcommand{\RIGHT}{\quad \Longrightarrow \quad}

\newif\ifdraft
\drafttrue

\usepackage{nicefrac}
\usepackage{authblk}

\title{Multi-Round Cooperative Search Games with Multiple Players\footnote{This work has received funding from the European Research Council (ERC) under the European Union's Horizon 2020 research and innovation program (grant agreement No 648032).}
}

\author[1]{Amos Korman}
\affil[1]{CNRS and University Paris Diderot, Paris, France\\  \texttt{amos.korman@irif.fr}}

\author[2]{Yoav Rodeh}
\affil[2]{Ort Braude College, Karmiel, Israel\\ \texttt{ yoav.rodeh@braude.ac.il}}

\newcommand{\prof}{\mathbb{P}}
\newcommand{\Pnash}{{\prof_\mathtt{nash}}}
\newcommand{\Pfinal}{\prof_\mathtt{final}}

\newcommand\Astar{A^\star}
\newcommand{\Aunif}{A_\mathtt{unif}}

\DeclareMathOperator{\success}{\mathtt{success}}

\DeclareMathOperator{\PoSS}{\mathtt{PoSS}}
\DeclareMathOperator{\PoA}{\mathtt{PoA}}
\DeclareMathOperator{\PoSA}{\mathtt{PoSA}}

\newcommand{\Cex}{{C_\mathtt{ex}}}
\newcommand{\Cshare}{{C_\mathtt{share}}}

\newcommand{\N}[1]{{\widetilde{#1}}}
\newcommand{\SM}[1]{\delta_{#1}}

\newcommand{\maxv}{\mathtt{maxv}}

\begin{document}
\date{}
\maketitle

\begin{abstract}
A treasure is placed in one of $M$ boxes according to a known distribution and $k$ searchers are searching for it in parallel during $T$ rounds. How can one incentivize selfish players so that the probability that at least one player finds the treasure is maximized?
We focus on {\em congestion policies}  $C(\ell)$ specifying the reward a player receives being one of the $\ell$ players that (simultaneously) find the treasure first.
We prove that the {\em exclusive policy}, in which  $C(1)=1$ and $C(\ell)=0$ for $\ell>1$, yields a {\em price of anarchy} of $(1-(1-{1}/{k})^{k})^{-1}$, which is the best among all symmetric reward policies. We advocate the use of symmetric equilibria, and show that besides being fair, they are highly robust to crashes of players. Indeed, in many cases, if some small fraction of players crash, symmetric equilibria remain efficient in terms of their group performance while also serving as approximate equilibria. 
\end{abstract}

\section{Introduction}
 Searching in groups is ubiquitous in  multiple contexts, including in the biological world, in human populations as well as on the internet \cite{JACM,Social-foraging,hills}.  In many cases there is some prior on the distribution of the searched target. Moreover, when the space is large, each searcher typically needs to inspect multiple possibilities, which in some circumstances can only be done  sequentially. 
 This paper introduces a game theoretic perspective to such multi-round treasure hunt searches, generalizing a basic collaborative Bayesian framework previously introduced in \cite{JACM}.

Consider the case that a treasure is placed in one of $M$ boxes according to a known distribution~$f$ and that $k$ searchers are searching for it in parallel during $T$ rounds, each specifying a box to visit in each round. Assume w.l.o.g.\ that the boxes are ordered such that lower index boxes have higher probability to host the treasure, i.e., $f(x)\geq f(x+1)$. 
We evaluate the group performance by
the {\em success probability}, that is, the probability that the treasure is found by at least one searcher.

If coordination is allowed, letting searcher $i$ visit box $(t-1)k + i$ at time $t$ will maximize success probability.  
However, as simple as this algorithm is, it is very sensitive to faults of all sorts. 
For example, if an adversary that knows where the treasure is can crash a searcher before the search starts (i.e., prevent it from  searching), then it can reduce the search probability to zero.

The authors of \cite{JACM} suggested the use of {\em identical non-coordinating} algorithms. In such scenarios all processors act independently, using no communication or coordination, executing the same probabilistic algorithm, differing only by the results of their coin flips. 
As argued in \cite{JACM}, in addition to their economic use of communication, identical non-coordinating algorithms enjoy inherent robustness to different kind of faults. For example, assume that there are $k+k'$ searchers, and that an adversary can fail up to $k'$ searchers. Letting all searchers run the best non-coordinating algorithm for $k$ searchers guarantees that regardless of which $\ell\leq k'$ searchers fail, the overall search efficiency is at least as good as the non-coordinating one for $k$ players. 
Of course, since $k'$ players might fail, any solution can only hope to achieve the best performance of $k$ players. As it applies to the group performance we term this property as {\em group robustness}. Among the main results in \cite{JACM} is  identifying a non-coordinating algorithm, denoted $A^\star$, whose expected running time is minimal among non-coordinating algorithms. Moreover, for every given $T$, if this algorithm runs for $T$ rounds, it also maximizes the success probability.

The current paper studies  the game theoretic version of this  multi-round search problem\footnote{We concentrate on the normal form version 
in which players do not receive any feedback during the search (except when the treasure is found in which case the game ends). In particular, we assume that players cannot communicate with each other.}. 
The setting of \cite{JACM} assumes that the searchers adhere fully to the instructions of a central entity. 
In contrast, in a game theoretical context, searchers are self-interested and one needs to incentivize them to behave as desired, e.g., by awarding those players that find the treasure first. For many real world contexts, the competitive setting is in fact the more realistic one to assume. Applications range from crowd sourcing \cite{papanastasiou2017crowdsourcing}, multi-agent searching on the internet \cite{boinc}, grant proposals \cite{kleinberg2011mechanisms}, to even contexts of animals \cite{IFD-review2} (see Appendix \ref{App:animals}).

In the competitive setting, choosing a good {\em rewarding policy}  becomes a problem in algorithmic mechanism design \cite{GameTheory}. 
Typically, a reward policy is evaluated by its {\em price of anarchy (PoA)}, namely, the  ratio between the  performances of the best collaborative algorithm and the worst  equilibrium~\cite{papa}. 
Aiming to both accelerate the convergence process to an equilibrium and obtain a preferable one, the announcement of the reward policy can be accompanied by  a proposition for players to play  particular strategies that form a profile  at equilibrium.

This paper highlights the benefits of suggesting  (non-coordinating) {\em symmetric equilibria} in such scenarios, that is, to suggest the same non-coordinating strategy to be used by all players, such that the resulting profile is at  equilibrium. 
This is of course relevant assuming that the {\em price of symmetric stability (PoSS)}, namely,  the ratio between the  performances of the best collaborative algorithm and the best symmetric equilibrium, is low.
Besides the obvious reasons of fairness and simplicity, from the perspective of a central entity who is interested in the overall success probability, we 
obtain the group robustness property mentioned above, by suggesting that the $k+k'$ players play according to the  strategy that is a symmetric equilibrium for $k$ players. 
Obviously, this group robustness is valid only provided that the players indeed play according to the suggested strategy. 
However, the suggested strategy is guaranteed to be an equilibrium only for $k$ players, while in fact, the adversary may keep some of the extra $k'$ players alive.
Interestingly, however, in many cases, a symmetric equilibrium for $k$ players also serves as an approximate equilibrium for $k+k'$ players, as long as $k'\ll k$.
As we show, this {\em equilibrium robustness} property is rather general, holding for a class of games, that we call  {\em monotonously scalable} games. 

\subsection{The Collaborative Search Game}\label{sec:model}

A treasure is placed in one of $M$ boxes according to a known  distribution $f$ and $k$ players are searching for it in parallel during $T$ rounds. Assume w.l.o.g.\ that $f(x)>0$ for every $x$ and that  $f(x)\geq f(x+1)$.

\subsubsection{Strategies}

An {\em execution} of $T$ rounds is a sequence of box visitations $\sigma=x(1),x(2),\dots x(T)$, one for each round $i\leq T$. We assume that a player visiting a box has no information on whether other players have already visited that box or are currently visiting it. Hence, a  {\em strategy} of a player is a 
probability distribution over the space of executions of $T$ rounds. Note that the probability of visiting a box $x$ in a certain round may depend on the boxes visited by the player until this round, but not on the actions of other players. A strategy is {\em non-redundant} if at any given round it always checks a box it didn't check before (as long as  there are such boxes). 

A {\em profile} is a collection of $k$ strategies, one for each player.  Special attention will be devoted to {\em symmetric profiles}. In such profiles all players play the same  strategy (note that their actual executions may be different, due to different probabilistic choices).

\subsubsection{Probability Matrix}

While slightly abusing notation, we shall associate  each  strategy $A$ with its \emph{probability matrix}, 
\[
A: \set{1, \ldots, M} \times \set{1, \ldots, T} \rightarrow [0,1],
\] where 
$A(x,t)$ is the probability that strategy $A$ visits $x$ for the first time at round $t$.
We also denote
$\N{A}(x,t)$ as the probability that $A$ does not visit $x$ by, and including, time $t$. That is, 
$\N{A}(x,t)= 1 - \sum_{s\leq t} A(x,t)$ and $\N{A}(x,0) = 1$. 
For convenience we denote by $\SM{x,t}$ the matrix of all zeros except $1$ at $x,t$. Its dimensions will be clear from context.

\subsubsection{Group Performance}

A profile is evaluated by its  \emph{success probability}, i.e.,   
the probability that at least one player finds the treasure by time $T$. 
Formally, let $\prof$ be a profile. Then, 
\begin{align*}
\success(\prof) & = \sum_x f(x)\B{1-\prod_{A \in \prof} \N{A}(x,T)}.
\end{align*}
The expected running time in the symmetric case, which is $\sum_x f(x) \sum_t  \N{A}(x,t)^k$, was studied in \cite{JACM}. That paper identified a strategy, denoted $\Astar$, that minimizes this quantity. In fact, it  does so by minimizing the term $\sum_x f(x) \N{A}(x, t)^k$ for each $t$ separately. 
Note that minimizing the  case $t=T$ is exactly the same as maximizing the success probability.
Thus, restricted to the case where all searchers use the same strategy, $\Astar$ simultaneously optimizes the success probability as well as optimizes the expected running time. A description of $\Astar$ is provided below.

\subsubsection{Algorithm \texorpdfstring{\boldmath{$\Astar$}}{Astar}}\label{sec:Astar}
We note that in \cite{JACM} the matrix of $\Astar$ is given, and then an algorithm is explicitly described that has its matrix (Section 4.3 in \cite{JACM}). We describe the matrix only, as its details are necessary for this paper. 
Denote $q(x) = f(x)^{-1/
(k-1)}$. For each $t$, 
\[
\N{\Astar}(x,t) = \min(1, \alpha(t) q(x))
,\]
where $\alpha(t) \geq 0$ is such that $\sum_x \N{\Astar}(x,t) = M-t$. 
Of course, once $\N{\Astar}$ is known, then so is $\Astar$: $\Astar(x,t) = \N{\Astar}(x,t-1) - \N{\Astar}(x,t)$.


\subsubsection{Congestion Policies}

A natural way to incentivize players is by rewarding those players that find the treasure before others.
A {\em congestion policy} $C(\ell)$ is a function specifying the reward that a player receives if it is one of $\ell$ players that (simultaneously) find the treasure for the first time. We assume that $C(1)=1$, and that $C$ is non-negative and non-increasing. 
Due to the fact that the policy in which $C\equiv 1$, i.e., $C(\ell) = 1$ for every $\ell$, is rather degenerate, we henceforth assume that there exists an $\ell$ such that $C(\ell)<1$.  We shall give special attention to the following policies. 

\begin{itemize}
\item
The \emph{sharing policy} is defined by $\Cshare(\ell)=1/\ell$, namely, the treasure is shared equally among all those who find it first. 
\item
The \emph{exclusive policy} is defined by $\Cex(1)=1$, and $\Cex(\ell)=0$ for $\ell>1$, namely,
the treasure is given to the first one that finds it exclusively; if more than one discover it, they get nothing\footnote{In the one round game, the exclusive policy yields a utility for a player that equals its marginal contribution to the social welfare, i.e., the success probability \cite{Welfare}. However, this is not the case in the multi-round game.}.
\end{itemize}
Note that we allow for $C(\ell) \cdot \ell$ to be greater than 1. For example, a competition may give 1\$ to all winners, as long as there are at most 3. If it would give 1\$ no matter how many winners there are, we would be in the policy $C\equiv 1$.

A {\em configuration} is a triplet $(C,f,T)$, where $C$ is a congestion policy, $T$ is a positive integer, and $f$ is a positive non-increasing probability distribution on $M$ boxes.

\subsubsection{Values, Utilities and Equilibria}

Let $(C,f,T)$ be a configuration. The \emph{value} of box $x$ at round $t$ when playing against\footnote{We stress that the player we look at, is external, in the sense that it is not part of the profile $\prof$, and in fact, does not participate in the game in rounds $1,2,\ldots,t-1$.} a profile $\prof$ is the expected gain from visiting $x$ at  round $t$. Formally,
\begin{align*}
v_\prof(x,t) 
& = 
f(x) \sum_{\ell=0}^{k-1} C(\ell+1) \prob{
\begin{array}{l}
\text{$x$ was not visited before time $t$, and} \\
\text{at time $t$ is visited by $\ell$ players of $\prof$}
\end{array}}
\\ & =
f(x) \sum_{\ell=0}^{k-1} C(\ell+1) 
\sum_{\stackrel{I \subseteq \prof}{|I| = \ell}} \prod_{A \in I} A(x,t) \prod_{A \not\in I} \N{A}(x,t)
.\end{align*}
We can then define {\em utility of $A$  in round $t$} and the {\em utility of $A$} as:
\begin{align*}
U_\prof(A,t) & =  \sum_{x} A(x,t) \cdot v_\prof(x,t), 
\\
U_\prof(A) & = \sum_{t} U_\prof(A,t).
\end{align*}
Here are some specific cases we are interested~in:

\begin{itemize}
\item For symmetric profiles, $v_A(x,t)$  denotes the value when playing against $k-1$ players playing~$A$. In this case,
\[
v_A(x,t) = f(x) \sum_{\ell=0}^{k-1}  C(\ell+1)
\binom{k-1}\ell A(x,t)^\ell \N{A}(x,t)^{k-\ell-1}.
\]
$U_A(B,t)$ and $U_A(B)$ are defined in an analogous way in this case. 
\item For the exclusive policy,
\[
v_\prof(x,t) = f(x) \prod_{A \in \prof} \N{A}(x,t)
.\]
\item For the exclusive policy in symmetric profiles,
\[
v_A(x,t) = f(x) \N{A}(x,t)^{k-1}
.\]
\end{itemize}

Denote $\prof^{-A}$ is the set of players of $\prof$ excluding $A$. 
A profile $\prof$ is a {\em Nash equilibrium} under a configuration, if for any $A\in\prof$ and any other  strategy $B$,
$U_{\prof^{-A}}(A) \geq U_{\prof^{-A}}(B).$
Similarly, a  strategy $A$ is called a \emph{symmetric equilibrium} if the profile $A^k$ consisting of all $k$ players playing according to $A$ is an equilibrium.
We also define approximate equilibria:
\begin{definition}
 For $\epsilon>0$, we say a profile $\prof$ is a $(1+\epsilon)$-equilibrium if
for every $A \in \prof$ and for every strategy $B$, $U_{\prof^{-A}}(B) \leq (1+\epsilon) U_{\prof^{-A}}(A)$.
\end{definition}

\subsubsection{A Game of Doubly-Substochastic Matrices}

Both the expressions for the success probability and utility solely depend on the values of the probability matrices associated with the  strategies in question. Hence we view all strategies sharing the same  matrix as {\em equivalent}. Note that a matrix does not necessarily correspond to a unique strategy, as illustrated by the following equivalent strategies, for which $A(x,t) = B(x,t) = 1/M$ for every $t\leq M$ and $0$ thereafter:

\begin{itemize}
\item  Strategy $A$ chooses uniformly at every round one of the boxes it didn't choose yet.
\item In the first round, strategy $B$ chooses a random $x\in \set{0,\ldots,M-1}$. Then, at each round $t \geq 1$ it visits box $(x+t\mod M) + 1$.
\end{itemize}
Matrices are much simpler to handle than strategies, and so we would rather think of our game as a game of probability matrices than a game of strategies. For this we need to characterize which matrices are indeed probability matrices of strategies.
Clearly, a probability matrix is non-negative. Also, by their definition, each row and each column sums to at most 1. Such a matrix is called \emph{doubly-substochastic}. It turns out that these conditions are sufficient. That is, every doubly-substochastic matrix is a probability matrix of some strategy. 

A doubly-substochastic matrix is a \emph{partial permutation} if it consists of only 0 and 1 values.
The following is a generalization of the Birkhoff - von Neumann theorem, proved for example in \cite{horn}. 
\begin{theorem} \label{thm:birkhoff2}
A matrix is doubly-substochastic iff it is a convex combination of partial permutations.
\end{theorem}
Furthermore, Birkhoff's construction \cite{birkhoff1946three} finds this decomposition in polynomial time, and guarantees it contains at most a number of terms as the number of positive elements of the matrix. The generalization of \cite{horn} does not change this claim significantly, as it embeds the doubly-substochastic matrix in a doubly-stochastic one which is at most 4 times larger.
\begin{corollary}\label{cor:birkhoff} \label{cor:mat2alg}
If matrix $A$ is doubly-substochastic then there is some strategy such that $A$ is its probability matrix. 
Furthermore, this strategy can be found in polynomial time, and is implementable as a polynomial algorithm.
\end{corollary}
\begin{proof}
First note that the claim is true if $A$ is a partial permutation. The  strategy in this case will be a deterministic  strategy, which may sometimes choose not to visit any box. In the general case,
Theorem \ref{thm:birkhoff2} states that there exist $\theta_1,\theta_2,\ldots,\theta_k$, such that $\sum_{i=1}^k \theta_i =1$ and partial permutations $A_1,A_2,\ldots,A_k$, such that $A=\sum_{i=1}^k \theta_i A_i$. As mentioned, each $A_i$ is the probability matrix of some  strategy $B_i$. Define the following  strategy $B$ as follows: with probability $\theta_i$ run  strategy $B_i$. 
Then, the probability matrix of $B$ is $\sum_i \theta_i A_i = A$, as required.
\end{proof}
We will therefore view our game as a game of doubly-substochastic matrices.

\subsubsection{Greediness}

Informally, a strategy is greedy at a round if its utility in this round is the maximum possible in this round.
Formally, given a profile $\prof$ and some strategy $A$, we say that $A$ is \emph{greedy} w.r.t.\ $\prof$ at time $t$ if for any strategy $B$ such that for every $x$ and $s<t$, $B(x,s) = A(x,s)$, we have $U_\prof(A,t) \geq U_ \prof(B,t)$.
We say $A$ is greedy w.r.t.\ $\prof$ if it is greedy w.r.t.\ $\prof$ for each $t\leq T$. A strategy $A$ is called {\em self-greedy} (or {\em sgreedy} for short) if it is greedy w.r.t.\ the profile with $k-1$ players playing~$A$.

\subsubsection{Evaluating Policies}

Let $(C,f,T)$ be a configuration. 
Denote by $\text{Nash}(C,f,T)$  the set of equilibria for this configuration, and by $\text{S-Nash}(C,f,T)$ the subset of symmetric ones. Let  ${\mathcal{P}(T)}$  be the set of all profiles of $T$-round strategies.
We are interested in the following~measures.

\begin{itemize}
\item 
The {\em Price of Anarchy (PoA)} is 
\[
\PoA(C,f,T):=
 \F{\max_{\prof\in {\mathcal{P}(T)}}{\success}(\prof)}{\min_{\prof \in \text{Nash}(C,f,T)}{\success}(\prof)}
.\]

\item
The {\em  Price of Symmetric Stability (PoSS)} is 
\[
\PoSS(C,f,T):= \F{\max_{\prof\in {\mathcal{P}(T)}}{\success}(\prof)}{\max_{A \in \text{S-Nash}(C,f,T)}{\success}(A^k)}
.\]

\item 
The {\em  Price of Symmetric Anarchy  (PoSA)} is 
\[
\PoSA(C,f,T)= \F{\max_{\prof\in {\cal{P}(T)}}{\success}(\prof)}{\min_{A \in \text{S-Nash}(C,f,T)}{\success}(A^k)}
.\]
\end{itemize}



\subsection{On the Difficulty of the Multi-Round Game}

The setting of multi-rounds poses several challenges that do not exist in the single round game. An important one is the fact that, in contrast to the single round game, the multi-round game is not a potential game. Indeed, being a potential game has several implications, and a significant one is that such a game always has a pure equilibrium. However, we show that multi-round games do not always have pure equilibria and hence are not potential games. 

\begin{claim}
The single-round game is a potential game, yet the multi-round game is not.
\end{claim}
\begin{proof}
For the single round, assume that $\prof$ is a deterministic profile, and let $x_\prof$ be the number of players that choose box $x$ in $\prof$. Denote,
\[\Phi(\prof)
=
\sum_x f(x) \sum_{\ell=1}^{x_\prof} C(\ell),\]
where we use the convention that $\sum_{\ell=1}^0 C(\ell)=0$.
If a player changes strategy and chooses (deterministically) some box $y$ instead of box $x$, then the change in its utility is $f(y) C(y_\prof+1) - f(x) C(x_\prof)$. This is also the change that $\Phi(\prof)$ sees. This extends naturally to mixed strategies, and so the single round game is a potential game. 

On the other hand, the multi-round game does not always have a pure equilibrium, and so is not a potential game. For example, the following holds for any policy $C$. Consider the case of a three boxes ($M=3$), two rounds ($T=2$), two players ($k=2$), and all boxes have an equal probability of holding the treasure ($f(x) = 1/3$). Note that $C(2) < 1$ since here $k=2$ and we assumed $C\not\equiv 1$.

Assume there is some deterministic profile that is at equilibrium, and w.l.o.g.\ assume player 1's first pick is box 1.
There are two cases:
\begin{enumerate}
\item
Player 1 picks it again in the second round. Player 2's strictly best response is to pick box 2 and then 3 (or the other way around). In this case, player 1 would earn more by first picking box 3 (box 2) and then box 1. In contradiction.
\item 
Player 1 picks a different box in the second round. W.l.o.g.\ assume it is box 2.
Player 2's strictly best response is to first take box 2 and then take box 3. However, player 1 would then prefer to start with box 3 and then box 1. Again a contradiction.
\qedhere
\end{enumerate}
\end{proof}

Another important difference
is that for policies that incur high levels of competition (such as the exclusive policy),  profiles that maximize the success probability are at equilibrium in the single round case, whereas they are not in the multi-round game.

In the single-round game, when $M\geq k$, the success probability is maximized when each player exclusively visits one box in $1,2,\ldots, k$ with probability 1. 
Under the exclusive policy, for example, such a profile is also at equilibrium.  
In fact, if $f(k)\geq f(1)/2$ then the same is true also for the sharing policy. 

For the multi-round setting, when $M\geq Tk$, an optimal scenario is also achieved by a deterministic profile, e.g., when player $i$ visits box $i+(t-1)k$ in round~$t$. However, this profile would typically not be an equilibrium, even under the exclusive policy.
Indeed, when $f(x)$ is strictly decreasing, player 2 for example, can gain more by stealing box $k+1$ from player 1 in the first round, then safely taking box 2 in the second round, and continuing from there as scheduled originally. This shows that in the multi-round game, the best equilibrium has only sub-optimal success probability.

\subsection{Our Results}

\subsubsection{Equilibrium Robustness}

We first provide a  simple, yet general, robustness result, that holds for symmetric (approximate) equilibria in a family of games, termed {\em monotonously scalable}. 
Informally, these are games in which 
the sum of utilities of players can only increase when more players are added, yet for each player, its individual utility can only decrease.  Our search game with the sharing policy is one such example. 
\begin{restatable}{theorem}{ThmRobustGeneral}\label{thm:robustgeneral}
Consider a symmetric monotonously scalable  game. If $A$ is a symmetric $(1+\epsilon)$-equilibrium for $k$ players, then it is an  $(1+\epsilon)(1+\ell/k)$-equilibrium when played by $k+\ell$ players. 
\end{restatable}
Theorem \ref{thm:robustgeneral} is applicable in fault tolerant contexts. Consider a monotonously scalable game with $k+k'$ players out of which at most $k'$ may fail. Let $A_k$ be a  symmetric (approximate) equilibrium designed 
for $k$ players and assume that its social utility is high compared to the optimal profile with $k$ players. 
The theorem implies that if players play $A_k$, then regardless of which    $\ell\leq k'$ players fail (or decline to participate), the incentive to switch 
strategy would be very small,
as long as  $k'\ll k$. Moreover, due to symmetry, if the social utility of the game is monotone, then the social utility of $A_k$ when played with $k$ players is guaranteed when playing with more.
Thus, in such cases we obtain both group robustness and equilibrium~robustness. 

\subsubsection{General Congestion Policies}
Coming back to our search game, we consider general policies, focus on symmetric profiles, and specifically, on the properties of sgreedy strategies. 
\begin{restatable}{theorem}{thmGreeeeed}\label{thm:greeeeed} 
For every policy $C$ there exists a  non-redundant sgreedy strategy. Moreover, all such strategies are equivalent and are symmetric $(1+C(k))$-equilibria.
\end{restatable}
When $C(k)=0$ this shows that a non-redundant sgreedy strategy is actually a symmetric equilibrium. We next claim  that this is the only symmetric equilibrium (up to equivalence).

\begin{restatable}{claim}{lemGreedyCk}\label{lem:greedyCk}
For any policy such that $C(k)=0$, 
all symmetric equilibria are equivalent. 
\end{restatable}

Theorem \ref{thm:greeeeed} is non-constructive because it requires calculating the inverse of non-trivial functions. Therefore, we resort to an approximate solution.

\begin{restatable}{theorem}{thmSgreedyComplete}\label{thm:SgreedyComplete}
Given $\theta>0$, there exists an algorithm that takes as input a configuration, and produces a symmetric $(1+C(k))(1+\theta)$-equilibrium. The algorithm runs in  polynomial time in $T$, $k$, $M$, $\log(1/\theta)$, $\log(1/(1 - C(k)))$, and $\log(1/f(M))$.
\end{restatable}

\subsubsection{The Exclusive Policy}
Recall that the exclusive policy is defined by $\Cex(1)=1$ and $\Cex(\ell)=0$ for every $\ell>1$. 
Recall also that $A^\star$ is the  symmetric strategy that, as established in \cite{JACM}, gives the optimal success probability among symmetric strategies (see Section \ref{sec:Astar}). 
We show that $A^\star$ is a non-redundant and sgreedy strategy w.r.t.\ the exclusive policy. Hence, Theorem \ref{thm:greeeeed}  implies the following.

\begin{restatable}{theorem}{thmAstar}\label{thm:Astar}
Under the exclusive policy, Strategy $\Astar$ of \cite{JACM} is a symmetric equilibrium.
\end{restatable}
Claim \ref{lem:greedyCk} together with the fact that $\Astar$ has the highest success probability among symmetric profiles,  implies that both the PoSS and the PoSA of $\Cex$ are optimal (and equal) on any $f$ and $T$ when compared to any other  policy.
The next theorem considers general equilibria.

\begin{restatable}{theorem}{thmNashismore}\label{thm:nashismore}
Consider the exclusive policy. For any profile $\Pnash$ at equilibrium and any symmetric profile $A$, $\success(\Pnash)\geq \success(A)$.
\end{restatable}
Observe that, 
as $A^\star$ is a symmetric  equilibrium, 
Theorem \ref{thm:nashismore} provides an alternative proof for the optimality of $A^\star$ (established in \cite{JACM}). Interestingly, this alternative proof is based on game theoretic considerations, which is a very rare approach in optimality proofs.

Combining Theorems \ref{thm:Astar} and \ref{thm:nashismore}, we obtain:
\begin{corollary}\label{cor:POA}
For any $f$ and $T$,
 $\PoA(\Cex, f,T) = \PoSA(\Cex, f ,T)$.
Moreover, 
for any policy $C$,  $\PoA(\Cex,f,T)\leq \PoA(C,f,T)$.
 \end{corollary}
At first glance 
the effectiveness of   $C_{ex}$ might not seem so surprising. Indeed, it seems natural that high levels of competition would  incentivize players to disperse.
However, it is important to note that $\Cex$ is not extreme in this sense, as one may allow congestion policies to also have negative values upon collisions.
 Moreover, one could potentially define more complex kinds of policies, e.g., policies that depend on time, and reward early finds more. However, the fact that $\Astar$ is optimal among all symmetric profiles combined with the fact that any symmetric policy has a symmetric equilibrium \cite{nash} implies that no symmetric reward mechanism can improve either the PoSS, the PoSA, or the PoA of the exclusive~policy.  

We proceed to show a tight upper bound on the PoA of $\Cex$. Note that as $k$ goes to infinity the bound converges to $e/(e-1)\approx 1.582$.
\begin{restatable}{theorem}{thmPrice}
\label{thm:price}
 For every $T$, $\sup_{f}\PoA(\Cex,f,T)  =(1-(1-1/k)^k)^{-1}.$
\end{restatable}

Concluding the results on the exclusive policy, we study the robustness of 
$\Astar$.
Let $\Astar_k$ denote  algorithm $\Astar$ when set to work for $k$ players. Unfortunately,  for any $\epsilon$, there are cases where $\Astar_k$ is not a $(1+\epsilon)$-equilibrium even when played by $k+1$ players.
However, as indicated below, $\Astar$ is robust to failures under reasonable assumptions regarding the distribution~$f$.

\begin{restatable}{theorem}{thmAstarRobust}\label{thm:AstarRobust}
 If $\F{f(1)}{f(M)}\leq  (1+\epsilon)^\F{k-1}{k'}$, then
$\Astar_k$ is a $(1+\epsilon)$-equilibrium when played by $k+k'$ players. 
\end{restatable}

\subsubsection{The Sharing Policy}

Another important policy to consider is the sharing policy. This policy naturally arises in some circumstances, and may be considered as a less harsh alternative to the exclusive one.  Although not optimal, it follows from Vetta \cite{vetta2002nash} that its PoA is at most 2  (see Appendix~\ref{apx:vetta}).  Furthermore, as this policy yields a monotonously scalable game, a symmetric equilibrium under it is also robust. Therefore, the existence of a symmetric profile which is both robust and has a reasonable success probability is guaranteed.  

Unfortunately,  we did not manage to 
find a polynomial algorithm that generates a symmetric equilibrium for this policy.  However, Theorem \ref{thm:SgreedyComplete} gives a symmetric $(1+\theta)(1+1/k)$-equilibrium in polynomial time for any $\theta > 0$. This strategy is also robust thanks to Theorem \ref{thm:robustgeneral}.
Moreover, the proof in  \cite{vetta2002nash} regarding the PoA can be extended to hold for approximate equilibria. 
In particular, if $\prof$ is some $(1+\epsilon)$-equilibrium w.r.t.\  the sharing policy, then for every $f$ and $T$,
$\success(\prof)\geq \R{2+\epsilon} \max_{\prof'\in \mathcal{P}(T)}\success(\prof')$. The proof of this claim appears in Appendix~\ref{apx:vetta}).

\subsection{Related Works}
Fault tolerance has been a major topic in distributed computing for several decades, and in recent years more attention has been given to these concepts in game theory
 \cite{Halpern1,Halpern2}.  
For example, Gradwohl and Reingold   studied conditions under which games are robust to faults,  showing that equilibria in anonymous games are fault tolerant if they are ``mixed enough''~\cite{OmerGame}.

Restricted to  a single round the search problem becomes a coverage problem, which has been investigated in several papers.  For example, Collet and Korman studied in \cite{collet}  (one-round) coverage while restricting attention to symmetric profiles only. The main result therein is that the exclusive policy yields the best coverage among symmetric profiles.
 Gairing \cite{Gairing} also considered the single round setting, but studied the optimal PoA of a more general family of  games called {\em covering games} (see also \cite{Welfare,Ramaswamy}). 
Motivated by policies for research grants, Kleinberg and Oren \cite{kleinberg2011mechanisms} considered a  one-round model similar to that in \cite{collet}. Their focus however was on pure strategies only.
The aforementioned papers give a good understanding of coverage games in the single round setting. As mentioned, however,  the multi-round setting studied here is substantially more complex than the  single-round setting. 

The multi-armed bandit problem \cite{gittins2011multi,slivkins2019introduction} is commonly recognized as a central exemplar of the exploration versus
exploitation trade-off. In the classic version, an agent (gambler) chooses sequentially between alternative arms (machines), each of which generates rewards according to an unknown distribution. The objective of the agent is to maximize the sum of rewards earned during the execution. In each round, the dilemma is between ``exploitation'', i.e, choosing the arm with the highest expected payoff, and ``exploration'', i.e, choosing an arm in order to 
acquiring knowledge that can be used to make better-informed decisions in the future. Most of the literature concerning the multi-armed bandit problem considers a single agent. This drastically differs from our setting that concerns multiple competing players. Furthermore, in our setting, information about the quality of boxes (arms) in known in advance, and there is no use in returning to a box (arm) after visiting it. 

The area of ``incentivized bandit exploration'', introduced by Kremer et al. \cite{kremer2014implementing},  studies incenticizing schemes in the context of the tradeoff between exploration and exploitation \cite{mansour2015bayesian,mansour2016bayesian,chen2018incentivizing}. 
In this version of multi-armed bandits problem, the action decisions are
controlled by multiple self-interested agents, based on recommendations given by a centralized algorithm. 
Typically, in each round, an agent arrives, chooses an arm to pull among several alternatives, receives a reward, and leaves forever. Each agent ``lives'' for a single round and therefore pulls the arm with the highest current
estimated payoff, given its knowledge. The centralized algorithm controls the flow of information, which in turn  can incentivize the agents to explore. 
 This is different from in our scenario in which agents remain throughout rounds, and know their entire history. Moreover, in our setting, the central entity controls only the congestion policy, and does not manipulate the information agents have throughout the execution. Within the scope of 
incentivized bandit exploration, related models have been studied in 
 \cite{crowdsourcing,frazier2014incentivizing} while considering time-discounted utilities, and in \cite{kannan2017fairness} studying the notion of fairness.
A extensive review of this line of research can be found in Chapter 11 of Slivkins
\cite{slivkins2019introduction}.

The settings of selfish routing, job scheduling, and congestion games
\cite{monderer1996potential,rosenthal1973congestion} all bear
similarities to the search game studied here, however, the social welfare measurements of  success probability or running time are very
different from the measures studied in these frameworks, such as
makespan or latency
\cite{makespan1,makespan2,GameTheory,makespan4}. 




\section{Robustness in Symmetric  Monotonously Scalable Games}\label{sec:robust}

Consider a symmetric game where the number of players is not fixed.
Let $U_\prof(A)$ denote the utility that a player playing $A$ gets if the other players play according to $\prof$ and let  $\sigma(\prof) = \sum_{A \in \prof} U_{\prof^{-A}}(A)$. 
We say that such a game is \emph{monotonously scalable} if:

\begin{enumerate}
\item Adding more players can only increase the sum of utilities, i.e., if $\prof \subseteq \prof'$ then $\sigma(\prof) \leq \sigma(\prof')$.
\item
Adding more players can only decrease the individual utilities, i.e., if $\prof \subseteq \prof'$ then for all $A \in \prof$, $U_{\prof^{-A}}(A) \geq U_{\prof'^{-A}}(A)$.
\end{enumerate}

\ThmRobustGeneral*
\begin{proof}
On the one hand by symmetry, 
\[
U_{A^{k+\ell-1}}(A) = \F{\sigma(A^{k+\ell})}{k+\ell} \geq
\F{\sigma(A^k)}{k+\ell}
,\]
where the last step is because $\sigma$ is non-decreasing.
On the other hand, if $B$ is some other strategy,
\[
U_{A^{k+\ell-1}}(B) \leq U_{A^{k-1}}(B) \leq (1+\epsilon)U_{A^{k-1}}(A) = 
(1+\epsilon)\F{\sigma(A^k)}{k}
.\]
The first inequality is because $U$ is non-increasing, and the second is because $A$ is a $(1+\epsilon)$-equilibrium for $k$ players.
Therefore, what a player can gain by switching from $A$ to $B$ is at most a multiplicative factor of  $(1+\epsilon)(k+\ell)/k = (1+\epsilon)(1 + \ell/k)$. 
\end{proof}
An example of such a game is our setting with the sharing policy. Note however,  that our game with the exclusive policy
does not satisfy the first property, as adding more players can actually deteriorate the sum of utilities.
Another example is a  generalization known as \emph{covering games}  \cite{Gairing}. This sort of game is the same as our game in a single-round version, except that each player chooses not necessarily one element, but a set of elements, from a  prescribed set of sets. Again, to be a monotonously scalable game, the congestion policy should be the sharing policy. Note that one may consider a multi-round version of these games, which will be monotonously scalable as well.

\section{General Policies}\label{sec:general}

A first useful observation is that if a box has some probability of not being chosen, then it has a positive value. This is clear from the definition of utility, from the fact that $C(1) = 1$ and  because $C$ is non-negative.
\begin{observation}\label{obs:non-zero}
If for all $A\in\prof$, $\N{A}(x,t) > 0$, then
$v_\prof(x,t) > 0$.
\end{observation}

\subsection{Non-Redundancy and Monotonicity}

A doubly-substochastic matrix $A$ is called \emph{non-redundant at time $t$} if $\sum_x A(x,t) = 1$. It is \emph{non-redundant} if it is non-redundant for every $t \leq M$.
In the algorithmic view, as $\sum_x A(x,t)$ is the probability that a new box is opened at time $t$, then a strategy's matrix is non-redundant iff it never checks a box twice, unless it already checked all~boxes.

\begin{observation} \label{obs:nonRedundantM}
If $A$ is non-redundant then for all $t \geq M$, and for all $x$, $\N{A}(x,t) = 0$.
\end{observation}

\begin{lemma}
\label{lem:non-redundant-0}
If a profile $\prof$ is at equilibrium and $\success(\prof)<1$ then every player is non-redundant. 
\end{lemma}
\begin{proof}
As $\success(\prof)<1$, there is some $x$
such that $\prod_{B\in\prof} \N{B}(x,T) > 0$, and so for all $B \in \prof$, $\N{B}(x,T) > 0$. Fix such an $x$. 
Assume that some player plays a redundant matrix $A$. This means that there is some time $s$ where 
$\sum_y A(y,s) < 1$. 
Define a new matrix 
$A' = A + \epsilon \SM{x,s}$, for a small $\epsilon > 0$. 
Taking it small enough will ensure that $A'$ is doubly-substochastic since (1) in column $s$, there is space because of the redundancy of $A$ at time $s$, and (2) in row $x$ there is space because $\sum_t A(x,t) = 1 - \N{A}(x,T) < 1$. 

Therefore our player can play according to $A'$ instead of $A$, and 
\[
U_{\prof^{-A}}(A') - U_{\prof^{-A}}(A) =
\epsilon \cdot v_{\prof^{-A}}(x,s)
.\]
Recall that by how we chose $x$, 
for all $B \in \prof^{-A}$, $\N{B}(x,T) > 0$, and as $\N{B}$ is weakly decreasing in $t$,  $\N{B}(x,s) > 0$. By Observation \ref{obs:non-zero},
$v_{\prof^{-A}}(x,s) > 0$, and so the utility strictly increases,
in contradiction to $\prof$ being at equilibrium.
\end{proof}

We will later see that in the symmetric case the condition in the lemma is not needed. However, the following example shows it is necessary in general.
Let $M=k$, $T>1$, and assume that for every $1\leq i\leq k$, player $i$ goes to box $i$ in every round. Under the exclusive policy, this strategy is an equilibrium, whereas each player is clearly redundant. 

Turning to the symmetric case, recall that $v_A(x,t)$ is the value of box $x$ at time $t$ when playing against $k-1$ players playing according to $A$. The following monotonicity lemmas hold for any congestion policy $C$. It basically says that, in a symmetric profile, if the probability of checking box $x$ at time $t$ is increased, then its value decreases at this time.
\begin{lemma}
\label{lem:alg-mono}
Consider two doubly-substochastic matrices $A$ and $B$, a box $x$ and a time $t$. If $A(x,t) > B(x,t)$, and for all $s<t$, $A(x,s) = B(x,s)$ then $v_A(x,t) < v_B(x,t)$. 
\end{lemma}
\begin{proof}
Let us consider a simple yet similar scenario: We toss $k-1$ i.i.d coins, each gives $1$ w.p.\ $p$, and $0$ otherwise. 
When the sum is $\ell$ the gain is $C(\ell +1)$, where $C$ is the congestion policy of the underlying game. The expected gain is then:
\[
\Psi(p) = 
\sum_{\ell=0}^{k-1} C(\ell+1) \binom{k-1}{\ell} p^\ell (1-p)^{k-1-\ell}
.\]
\begin{claim}
$\Psi$ is strictly decreasing in $p$.
\end{claim}
\begin{proof}
Denote by $R(\ell) = \binom{k-1}\ell p^\ell (1-p)^{k-1-\ell}$ the probability that the sum is $\ell$, and by 
$Q(\ell) = \sum_{s=0}^\ell R(s)$ the probability that the sum is at most $\ell$.
\begin{align*}
\Psi(p) 
& = 
C(1)R(0) + C(2)R(1) + \ldots + C(\ell)R(\ell-1)
\\ & = 
C(1)Q(0) + C(2)(Q(1) - Q(0)) + \ldots + C(k)(Q(k-1) - Q(k -2)) 
\\ & =
(C(1) - C(2))Q(0) + \ldots + (C(k-1) - C(k))Q(k - 2) + C(k)Q(k-1).
\end{align*}
By the interpretation of $Q(\ell)$, it is clear that it is strictly decreasing in $p$, except for $Q(k-1)$ which is always equal to 1.
By the properties of congestion policies, for all $\ell \leq k-1$,  $C(\ell) - C(\ell+1) \geq 0$. 
Moreover, since $C\not\equiv 1$, then at least one $C(\ell) - C(\ell+1)$ is strictly positive, concluding the claim.
\end{proof}
Now, back to our (symmetric) game, recall that the value is defined as:
\[
v_A(x,t) 
= 
f(x) \sum_{\ell=0}^{k-1}  C(\ell+1)
\binom{k-1}\ell A(x,t)^\ell \N{A}(x,t)^{k-\ell-1}.
\]
Note that:
\[
\N{A}(x,t) + A(x,t) = \N{A}(x,t-1)
.\]
And so:
\begin{align*}
v_A(x,t) & =
f(x) \cdot \B{\N{A}(x,t-1)}^{k-1} \sum_{\ell=0}^{k-1}  C(\ell+1)
\binom{k-1}\ell \BF{A(x,t)}{\N{A}(x,t-1)}^\ell \BF{\N{A}(x,t)}{\N{A}(x,t-1)}^{k-\ell-1}.
\\
& = 
f(x) \cdot \B{\N{A}(x,t-1)}^{k-1} 
\Psi\BF{A(x,t)}{\N{A}(x,t-1)}.
\end{align*}
By the conditions of the lemma, $A$ and $B$ behave the same w.r.t $x$, up to, and including time $t-1$, and therefore: \[\N{A}(x,t-1) = \N{B}(x,t-1).\] 
Using this to compare $v_A(x,t)$ and $v_B(x,t)$ according to the equation above, and using the fact that $\Psi$ is strictly decreasing, we get the result.
\end{proof}

\begin{lemma}
\label{lem:mono}
Let $A$ be doubly-substochastic. For every $x$ and $t$, 
$v_A(x,t+1) \leq v_A(x,t)$. Moreover, if $A(x,t+1) > 0$ then
the inequality is strict.
\end{lemma}
\begin{proof}
First Assume $A(x,t+1) = 0$.
The value at round $t+1$ is:
\[
v_A(x,t+1) 
= 
f(x) \sum_{\ell=0}^{k-1} C(\ell+1) \binom{k-1}\ell A(x,t+1)^\ell \N{A}(x,t+1)^{k-\ell-1}
= 
f(x) C(1) \N{A}(x,t+1)^{k-1}
\]
On the other hand,
\[
v_A(x,t) 
= 
f(x) \sum_{\ell=0}^{k-1} C(\ell+1) \binom{k-1}\ell A(x,t)^\ell \N{A}(x,t)^{k-\ell-1}
\geq
f(x) C(1) \N{A}(x,t)^{k-1}
,\]
because $C$ is non-negative. As $\N{A}(x,t) = \N{A}(x,t+1)$, we get that $v_A(x,t+1) \leq v_A(x,t)$, as required.

Next, assume $\N{A}(x,t) > 0$. Consider  strategy $A$ and let $B$ be the same as $A$ except that for every $s > t$, $B(x,s)=0$. 
By the above, $v_B(x,t+1) \leq v_B(x,t)=v_A(x,t)$. By Lemma \ref{lem:alg-mono}, \[v_A(x,t+1)<v_B(x,t+1),\] and we conclude. 
\end{proof}

Using the above, we prove a stronger result than Lemma \ref{lem:non-redundant-0} for the symmetric case:
\begin{restatable}{lemma}{LemSymNonRedundant}
\label{lem:sym-non-redundant}
If $A$ is a symmetric equilibrium then it is non-redundant.
\end{restatable}
\begin{proof}
Because of Lemma \ref{lem:non-redundant-0} it is sufficient to consider only the case where $\success(A) = 1$.  Let $T' = \min\{M, T\}$,  
and assume by contradiction that $A$ is redundant. Thus there is some $t \leq T'$ where $\sum_x A(x,t) < 1$. 
Hence, 
\[
\sum_{s\leq T'} \sum_x A(x,s) < T'.
\]
Therefore, there is some $x$ such that $\sum_{s\leq T'} A(x,s) < 1$ and so $\sum_{s\leq t} A(x,s) < 1$. 
As $\success(A)=1$, there is some $t' > t$ such that $A(x,t') > 0$. 
Define: 
\[
A' = A + \epsilon(\SM{x,t} - \SM{x,t'}).
\]
Taking $\epsilon > 0$ small enough, $A'$ is doubly-substochastic. Also, 
\[
U_A(A') - U_A(A) = \epsilon(v_A(x,t)- v_A(x,t')),
\]
which is strictly positive
by Lemma \ref{lem:mono}.
This contradicts the fact that $A$ is an equilibrium.
\end{proof}

\subsection{Greedy Strategies}

The following notation will aid us in what follows.
For a profile $\prof$ and a strategy $A$, denote:
\[
v_{\prof}(A,t) = \min_x\stset{v_\prof(x,t)}{A(x,t) > 0}.
\]
It is the value of the least valuable box that strategy $A$ has a positive probability of choosing.
Clearly, if $A$ is non-redundant then $U_\prof(A,t) \geq v_\prof(A,t)$. For symmetric profiles, we denote $v_A(t)$ as shorthand for $v_{A^{k-1}}(A, t)$.

A simple and useful observation is that if a greedy strategy has some probability of leaving a box unopened by (and including) time $t$, then the box's value cannot be too large at this time. Specifically, it is at most that of any box that has a positive probability of being chosen at the same time. 
\begin{observation} \label{obs:greedy}
Let $A$ be greedy at time $t$ w.r.t.\ profile $\prof$. If $\N{A}(x,t) > 0$ then
$v_\prof(x,t) \leq v_\prof(A,t)$. \end{observation}
\begin{proof}
Assume otherwise, that is, $v_\prof(x,t) > v_\prof(A,t)$.
Take some $y$ s.t.\ $v_\prof(y,t) = v_\prof(A,t)$, and let $B$ coincide with $A + \epsilon(\delta_{x,t} - \delta_{y,t})$ for the first $t$ columns, and $0$ for the rest. Taking $\epsilon > 0$ small enough makes $B$ doubly-substochastic. But
\[
U_\prof(B,t) - U_\prof(A,t) = 
\epsilon\B{ v_\prof(x,t) - v_\prof(y,t)
} > 0,
\]
contradicting the greediness of $A$.
\end{proof}

\subsubsection{An Equivalent Definition of Greediness}

\begin{lemma}\label{lem:greedydef}
A non-redundant strategy $A$ is greedy w.r.t.\ $\prof$ at time $t$ iff for every $x$ and $y$, if $A(x,t)>0$ and $v_\prof(x,t)<v_\prof(y,t)$ then $\N{A}(y,t)=0$. 
\end{lemma}
\begin{proof}
If $A$ is greedy at time $t$, then assume by contradiction, that there are some $x$ and $y$ where $A(x,t)>0$, $v_\prof(x,t) < v_\prof(y,t)$, and yet $\N{A}(y,t) > 0$. By Observation \ref{obs:greedy}, 
$v_\prof(y,t) \leq v_\prof(A,t) \leq v_\prof(x,t)$, in contradiction.

Conversely, if the condition in the lemma is true, then let $B$ be some matrix as required by the definition of greediness. 
We partition the set of boxes into two sets. {\em High-value} boxes are those $x$ such that $v_\prof(x,t)>v_\prof(A,t)$, and {\em low-value} boxes are all the rest. By the condition in the lemma, we know that for every high-value box $x$, $\N{A}(x,t)=0$, and hence $A(x,t)$ is set to its maximal quantity $\N{A}(x,t-1)$. Therefore, for each such box, its contribution to the utility of $A$ is at least as large as its contribution to $B$'s utility, i.e.,
\[
A(x,t)v_\prof(x,t)\geq B(x,t)v_\prof(x,t).
\]
The total contribution of high-value boxes to the utility of $A$ is therefore also 
at least as large as their contribution to $B$'s utility. 

On the other hand, the total contribution of low-value boxes to the utility of $B$ is at most $v_\prof(A,t)$ times the probability that $B$ goes to a low-value box in round $t$. Since $v_\prof(A,t)$ is the lowest value that $A$ checks, and since $A$ is non-redundant, we get that in total, 
\[
\sum_x A(x,t) v_\prof(x,t) \geq \sum_x B(x,t) v_\prof(x,t),
\]
as required. 
\end{proof}

\subsubsection{Existence and Uniqueness of SGreedy Strategies}

Recall that a strategy is denoted \emph{sgreedy} if it greedy w.r.t\ $k-1$ copies of itself.

\begin{lemma}\label{lem:existence}
For every policy $C$ there exists a  non-redundant sgreedy strategy $A$. Moreover, all such strategies are equivalent. 
\end{lemma}

The proof of existence constructs a sgreedy matrix $A$ inductively on $t$.
Intuitively, for time $t$, we examine a candidate value $w$ and find the $t$-th column of $A$ which satisfies $w=v_A(t)$, where we recall that $v_A(t)
= \min_x\stset{v_A(x,t)}{A(x,t) > 0} 
$.
This is done according to Lemma \ref{lem:greedydef}, which gives conditions for a non-redundant $A$ to be  sgreedy. 
The boundary cases are easy.  Specifically, if by setting $A(x,t) = 0$ we get $v_A(x,t) \leq w$ then we keep $A(x,t) = 0$. This is because adding more probability to $x$ would only reduce  the value $v_A(x,t)$, and we cannot have a box in the support that is below $v_A(t)$. Similarly, if by setting $A(x,t)=\N{A}(x,t-1)$ we have $v_A(x,t) \geq w$, then we keep 
$A(x,t)=\N{A}(x,t-1)$, which means that the box is filled. Otherwise, we need to show that there is a value for $A(x,t)$ that  gives $v_A(x,t)=w$. The existence of such a value is shown using the fact that $v_A(x,t)$ is a monotone and continuous function.

However, the sum of these $A(x,t)$'s over the $x$'s needs to be 1 for $A$ to be doubly-substochastic and non-redundant.
We therefore prove that by continuity, there exists a $w$ giving the desired sum of 1.
\begin{proof}\emph{(of Lemma \ref{lem:existence})}
As our construction will be non-redundant, for any $t > M$, all boxes will already be checked, and so setting $A(x,t) = 0$ is fine.
Fix a time $t\leq \min\{T,M\}$.  
Assume that we already defined all the $A(x,s)$ for $s<t$ such that $A$ is non-redundant and sgreedy for all such $s$. We will find values for $A(x,t)$ that are the same, and thus prove the lemma by induction.

Consider $v_A(x,t)$ as a function of $A(x,t)$, with $\N{A}(x,t-1) = 1 - \sum_{s<t} A(x,s)$ fixed. By Lemma~\ref{lem:mono}, it is a strictly decreasing function, it is continuous, and is defined between $0$ and $\N{A}(x,t-1)$. This means it has a continuous inverse function with this range. Denote by $A_w(x,t)$ the extended inverse. That is, given $w$, it is the $A(x,t)$ such that setting it gives $v_A(x,t) = w$ if such $A(x,t)$ exists. If this is not possible, it means either $w$ is too large, and so $A_w(x,t) = 0$, or too small, and then $A_w(x,t) = \N{A}(x,t-1)$.

Next, we choose $w$ so that $\sum_x A_w(x,t)=1$. This would be needed to ensure that $A$ is doubly-substochastic and non-redundant. For that, we note that $\sum_x A_0(x,t) = \sum_x \N{A}(x,t-1) \geq 1$, as $t -1 < M$. On the other hand, $1$ is larger than all $f(x)$, and so is an upper bound on the values boxes can take. Therefore, $\sum_x A_1(x,t) = 0$. 
By the fact that the $A_w(x,t)$ are a continuous function of $w$, this means there is some $w$ such that $\sum_x A_w(x,t)=1$. 
Moreover, it is easy to see (no matter what $w$ is) that these $A_w(x,t)$'s satisfy the greediness condition of Lemma \ref{lem:greedydef}. 
In other words, setting $A(x,t) = A_w(x,t)$, extends $A$ to time $t$ in that it is sgreedy and non-redundant, as required.

Next we prove uniqueness. 
Assume by contradiction that $A \neq B$ and both are sgreedy and non-redundant. Take the first column $t \geq 1$ where they differ. 
As the matrices are non-redundant, each column sums to 1, and therefore there is some $x$ where 
$A(x,t) > B(x,t)$ and some $y$ where 
$B(y,t) > A(y,t)$, which
by Lemma \ref{lem:alg-mono}, imply that $v_A(x,t) < v_B(x,t)$ and 
$v_B(y,t) < v_A(y,t)$. Now, 
\[
v_A(t) \leq v_A(x,t) < v_B(x,t) \leq v_B(t)
,\]
where the last inequality is because
$B(x,t) < A(x,t) \leq \N{A}(x,t-1) = \N{B}(x,t-1)$, and so Observation~\ref{obs:greedy} applies.
On the other hand, in the same manner,
\[
v_B(t) \leq v_B(y,t) < v_A(y,t) \leq v_A(t),
\] in contradiction.
\end{proof}

\subsubsection{SGreedy Strategies are Approximate Equilibria}

Now we can prove one of our main theorems:
\thmGreeeeed*
Let us first define:
\begin{definition} \label{def:fill}
We say that $A$ {\em fills} box $x$ at round $t$ if $A(x,t)>0$ and $\N{A}(x,t)=0$. 
\end{definition}
Intuitively, the proof that a non-redundant sgreedy strategy is an approximate equilibrium proceeds as follows. Clearly, for a given round $t$, any player $B$ playing against $k-1$ copies of $A$ cannot gain more than  $\maxv(t)=\max_x  v_A(x,t)$. 
Furthermore, for rounds where no box is filled by $A$, it really can't do any better than $A$, since, as shown in Item 1 in the proof below, the utility of $A$ in these cases is equal to
$\maxv(t)$. 

We further show that this is also true for the special case of the first round, no matter whether a box is filled or not at that round (Item 2).
For the other rounds, it turns out that $A$'s utility is always at least as large as the maximum value of any later round (Item 3).

Overall, using a telescopic argument, the sum of maximum values over the rounds is at most that of $A$'s utilities, plus the max value at the first round greater than 1 where a box is filled.
To upper bound this additive term, we prove (Item 4) that the value of a box when it is filled is at most $C(k)$ times its value at the first round, and so at most $C(k)$ times the total utility of $A$.

\begin{proof}\emph{(of Theorem \ref{thm:greeeeed})}
The existence of a strategy $A$ that is non-redundant and sgreedy is what was proved in Lemma \ref{lem:existence}. We prove here that such a strategy is a $(1+C(k))$-equilibrium. Consider a strategy $B$. We compare the utility of $B$ to that of $A$ when both play against $k-1$ players playing $A$. By non-redundancy, all of $v_A(x,t)$ are $0$ when $t>M$, and so we can assume $T\leq M$.

Since the utility of $B$ in any round $t$ is a convex combination of $v_A(x,t)$, we have
$U_A(B,t) \leq \maxv(t)$.
The following four claims hold for any round $t$:

\begin{enumerate}
\item 
If $A$ does not fill any box at round $t$ then 
$U_A(A, t) = \maxv(t)$. \\ This is because $U_A(A,t)$ is a convex combination of $v_A(x,t)$ for the boxes where $A(x,t) > 0$, which by the characterization of greediness in Lemma \ref{lem:greedydef}, all have the same value at time $t$.
\item
$U_A(A, 1) = \maxv(1)$. \\
Why? if no box is filled in round 1,  then Item 1 applies. Otherwise, for some box $x$, $A(x,1) = 1$, and all other boxes have $A(\cdot, 1) = 0$. The result follows again by Lemma \ref{lem:greedydef}.
\item
For any $s<t$, $U_A(A,s) \geq \maxv(t)$. \\
We prove this by showing that for every $x$, 
$U_A(A,s) \geq v_A(x,t)$.
If $v_A(x,t) = 0$, then the claim is clear. 
Otherwise, $A(x,t) > 0$ or $\N{A}(x,t) > 0$ or both. Either way, $\N{A}(x,s) > 0$. Therefore, as $A$ is sgreedy, for every $y$ such that $A(y,s)>0$, 
$
v_A(y,s) \geq v_A(x,s) \geq v_A(x,t)
$. The last inequality follows from monotonicity, i.e., Lemma \ref{lem:mono}. As $v_A(A,s)$ is a convex combination of such $y$'s we conclude.
\item
If $A$ fills box $x$ at time $t>1$ then $v_A(x,t) \leq C(k) v_A(x,1)$.\\
To see why, first note that
$
v_A(x,1) \geq f(x) C(1) \N{A}(x,1)^{k-1} =
f(x) \N{A}(x,1)^{k-1}$.
On the other hand, since $\N{A}(x,t) = 0$, we have
$
v_A(x,t) = f(x) C(k) A(x,t)^{k-1} \leq
f(x) C(k) \N{A}(x,1)^{k-1}$,
because $A(x,t) \leq \N{A}(x,t-1) \leq \N{A}(x,1)$. 
Combining the above two inequalities gives the result.
\end{enumerate}
Denote by $X_\text{fill}$ the set of rounds for which there is at least one box $x$ that is filled by $A$. Let $X_\text{no-fill}$ be the rest of the rounds, except for $t=1$ which is in neither. 
\[
U_A(B) 
\leq \sum_t \maxv(t)
=
\sum_{t \in X_\text{no-fill} \cup \set{1}} \maxv(t)
+
\sum_{t \in X_\text{fill}} \maxv(t).
\]
We want to show that this is not much larger than $U_A(A) = \sum_t U_A(A, t)$.
By Items 1 and 2 above, for each $t \in X_{\text{no-fill}} \cup \set{1}$, $U_A(A,t) = \maxv(t)$.
Therefore, the first sum is exactly $\sum_{t \in X_\text{no-fill} \cup \set{1}} U_A(A, t)$.

Regarding the second sum, denote $X_\text{fill} = \set{t_0, t_1, \ldots, t_k}$ in order. 
By Item 3, for every $i \geq 1$, 
\[
\maxv(t_i) \leq  U_A(A, t_{i-1})
.\]
Therefore,
\[
\sum_{t \in X_\text{fill}} \maxv(t)
\leq
\maxv(t_0) + \sum_{t \in X_\text{fill} \setminus \set{t_k}} U_A(A,t) 
\leq
\maxv(t_0) + \sum_{t \in X_\text{fill}} U_A(A,t) 
.\]
To sum up,
\[
U_A(B) 
\leq 
\sum_t U_A(A,t) + \maxv(t_0)
= U_A(A) + \maxv(t_0)
\]
Finally, we bound the effect of adding $\maxv(t_0)$. Since $t_0 \in X_\text{fill}$, then $t_0 > 1$.  
By Items 4 and 2:
\[
\maxv(t_0) 
= \max_x v_A(x,t_0) 
\leq \max_x C(k) v_A(x,1) 
= C(k) U_A(A,1)
\leq
C(k) U_A(A)
.\qedhere\]
\end{proof}
In Appendix \ref{exm:non-eq} we provide an example showing that in the sharing policy, a non-redundant sgreedy strategy is not necessarily at equilibrium. On the other hand, it is worth noting that for any policy, the existence of a symmetric equilibrium follows from \cite{nash}, and
for $C(k)=0$ we can get a full characterization of such equilibria:

\subsubsection{Symmetric Equilibria are SGreedy when \texorpdfstring{\boldmath{$C(k)=0$}}{C(k)=0}}

Theorem \ref{thm:greeeeed} tells us that
if $C(k)=0$ then there exists a non-redundant sgreedy strategy, and it is a symmetric equilibrium. We next claim that this is the only symmetric equilibrium up to equivalence.

\lemGreedyCk*
\begin{proof}
Let $A$ be some symmetric equilibrium. We will show it is non-redundant and sgreedy, and so by Theorem \ref{thm:greeeeed} we get that it is unique up to equivalence. By Lemma \ref{lem:sym-non-redundant}, $A$ is non-redundant, and we may assume that $T \leq M$. 
If $A$ is not sgreedy, then let $x,y$ and $t$ consist of a counter example, i.e.,
$A(x,t) > 0$ and yet $v_A(x,t) < v_A(y,t)$, where $\N{A}(y,t) > 0$.
There are two cases, and in both we will construct some doubly-substochastic matrix $B$ such that $U_A(B) > U_A(A)$, thus contradicting the claim that $A$ is a symmetric equilibrium.\\

\noindent {\bf Case 1.} $\sum_{s\leq T} A(y,s) < 1$.
In this case,  
let $B = A + \epsilon(\SM{y,t} - \SM{x,t})$.
Taking a small enough $\epsilon$  ensures that $B$ is a doubly-substochastic matrix, since (1) $y$'s row sums to strictly less than $1$ in $A$,  (2) the sum of column $t$ is the same as in $A$, and 
(3) $x$'s row only decreased in value compared to $A$.
Lastly, \[U_A(B) - U_A(A) = \epsilon (v_A(y,t) - v_A(x,t)) > 0.\] 

\noindent{\bf Case 2.} $\sum_{s\leq T} A(y,s) = 1$. This means that there is some first $t'$, where $\N{A}(y,t') = 0$. As $\N{A}(y,t)>0$, we have $t'>t$. Now:
\[
v_A(y,t') = f(y) \sum_{\ell=0}^{k-1} C(\ell+1) \binom{k-1}\ell A(y,t')^\ell \N{A}(y,t')^{k-\ell-1} =
f(y) C(k) A(y,t')^{k-1} = 0
,\]
because $C(k) = 0$.
Define: 
\[
B = A + \epsilon(-\SM{x,t} + \SM{y,t} - \SM{y,t'})
.\]
Taking a small enough $\epsilon$, as both $A(x,t)$ and $A(y,t')$ are strictly positive, $B$ is non-negative.
As $A$ is doubly-substochastic, $B$ is doubly-substochastic as well. 
\[
U_A(B) - U_A(A) = \epsilon(v_A(y,t) - v_A(x,t) - v_A(y,t'))
.\]
However, $v_A(y,t') = 0$, and $v_A(y,t) > v_A(x,t)$, and so this quantity is strictly positive.
\end{proof}

Interestingly, this result does not extend to non-symmetric profiles even for the exclusive policy, as is demonstrated by the following example of a non-greedy non-redundant equilibrium.
Consider three players and two rounds. $f(1) = f(2) = f(3) = (1-\epsilon)/3, f(4) = \epsilon$, for some small positive~$\epsilon$.
 Player 1 plays first 4 and then 1. Player 2 plays 2 and then 3, and player 3 plays 3 and then 2. This can be seen to be an equilibrium, yet player 1 is not greedy.

\subsection{Constructing Approximate Equilibria}

Theorem \ref{thm:greeeeed} is non-constructive in general, and does not allow us to actually find the desired sgreedy strategy in polynomial time.
Its possible to approximate the sgreedy strategy, step by step, but the errors accumulate exponentially and so the utilities for the players become skewed. Therefore, unless exponential time is invested, the resulting strategy is far from being an equilibrium. 

We therefore forget the true sgreedy strategy, and define   approximate greediness and non-redundancy. This gives a cheaper way of reaching an approximate equilibrium.
Our goal here is to prove:
\thmSgreedyComplete*

\begin{definition}\label{def:eps-sgreedy}
We say $A$ is {\em $\epsilon$-sgreedy} at time $t$ if whenever $A(x,t) > 0$ and $v_A(y,t) > v_A(x,t) + \epsilon$ then $\N{A}(y,t) = 0$. It is $\epsilon$-sgreedy if it is $\epsilon$-sgreedy for each time $t$.
\end{definition}
\begin{definition}
Let $0\leq \delta\leq 1$. We say $A$ is {\em $\delta$-redundant} if for every $t \leq M$, $\sum_x A(x,t) \geq 1 - \delta$. \end{definition}
We then prove the following two lemmas (in subsections  \ref{sec:SgreedyUse} and \ref{sec:lemGreedyCreate} below)
\begin{restatable}{lemma}{lemSgreedyUse}\label{lem:sgreedyUse}
If $A$ is $\delta$-redundant and  $\epsilon$-sgreedy then for any $B$, 
\[
U_A(B) \leq \F{1+C(k)}{1-\delta}U_A(A) + (T+1)\epsilon.
\]
\end{restatable}
In a sense, this lemma is a generalization of Theorem \ref{thm:greeeeed} which says that for $\delta=0$ and $\epsilon=0$, $A$ is a $(1+C(k))$-equilibrium. 

\begin{restatable}{lemma}{lemGreedyCreate}\label{lem:greedyCreate}
There is an algorithm that 
given a configuration, $\epsilon > 0$ and $\delta > 0$,
finds a matrix that is $\delta$-redundant and $\epsilon$-sgreedy.
The algorithm runs in polynomial time in parameters  
$T$, $k$, $M$, $\log(1/\epsilon)$, $\log(1/\delta)$, $\log(1/(1 - C(k)))$, and $\log(1/f(M))$.
\end{restatable}
Using these two lemmas, the theorem is not difficult to prove.
\begin{proof} \emph{(of Theorem \ref{thm:SgreedyComplete})}
First, if $M=1$, the strategy at equilibrium would be to pick box $1$ at $t=1$. We therefore assume $M \geq 2$.
Set $\epsilon = \theta/2(T+1) \cdot f(2)/2^{k+1}$  and $\delta= \F{\theta/2}{1 + \theta/2}$. 
Hence, $1/(1 - \delta) = 1 + \theta/2$. 
We use Lemma \ref{lem:greedyCreate} to construct a matrix $A$ which is $\delta$-redundant and $\epsilon$-sgreedy. As $\delta = \Theta(\theta)$, this construction takes polynomial time as required.
Also, according to Corollary  \ref{cor:birkhoff}, this strategy is implementable by a polynomial algorithm in the size of the matrix. 
According to Lemma \ref{lem:sgreedyUse}, for any  strategy $B$,
\begin{equation}\label{eq:ub}
U_A(B) 
\leq (1+C(k)) \B{1 + \F{\theta}{2}} U_A(A) + \F{\theta}{2} \cdot \F{f(2)}{2^{k+1}} 
.\end{equation}
As $M \geq 2$, 
then either $A(1,1)\leq 1/2$ or $A(2,1)\leq 1/2$. As $v_A(x,1) \geq f(x) C(1) \N{A}(x,1)^{k-1} =
f(x) \N{A}(x,1)^{k-1}$, we get that $\maxv(1)\geq f(2) / 2^{k-1}$. Since $\epsilon < f(2) / 2^{k}$, and $\delta$ can be assumed to be at most $1/2$, we get that $U_A(A) \geq f(2) / 2^{k+1}$. Therefore,
\[
\F{\theta}{2} \cdot \F{f(2)}{2^{k+1}} 
\leq (1+C(k))\F{\theta}{2}U_A(A)
,\]
which combined with Eq.~\eqref{eq:ub} gives the result.
\end{proof}

\subsubsection{Approximate Redundant and Sgreedy implies Approximate Equilibrium}\label{sec:SgreedyUse}


\lemSgreedyUse*
\begin{proof}
The proof follows the same steps as that of Theorem \ref{thm:greeeeed}, except it deals with the approximate redundancy and sgreediness.
Consider a strategy $B$. We compare the utility of $B$ versus that of $A$ when both play against $k-1$ players playing $A$.

By Lemma \ref{lem:mono}, $v_A(x,t)$ is non-increasing in $t$, and so 
we can assume w.l.o.g.\ that $B$ is non-redundant. Indeed, there is no point in postponing opening a box until a later time. In particular, we can assume that $B$ does not choose any box beyond time $M$. Therefore, we shall let $A$ run at most $M$ rounds, and prove the result assuming this new $A$ and $T\leq M$. This can only decrease the r.h.s.\ in the statement of the lemma, and thus is enough.

Denote $\maxv(t) = \max_x v_A(x,t)$.
 Since the utility of $B$ in any round $t$ is a convex combination of $v_A(x,t)$, we have:
\begin{equation}\label{eq:utiB}
U_A(B,t) \leq \maxv(t).
\end{equation}
Recall from Definition \ref{def:fill}, that $A$ {\em fills} box $x$ at round $t$ if $A(x,t)>0$ and $\N{A}(x,t)=0$. The following claim 
corresponds to Item 1 in the proof of Theorem \ref{thm:greeeeed}.
\begin{claim}\label{clm:nothisand that}
If $A$ does not fill any box at round $t$ then 
$U_A(A, t) \geq (1-\delta)(\maxv(t) - \epsilon)$. \end{claim}
\begin{proof}
We first argue that under the assumption of the claim, for every $x$ such that $A(x,t)>0$, we have $v_A(x,t)\geq \maxv(t) - \epsilon$. 
Indeed, assume by contradiction that $v_A(x,t) < v_A(y,t) - \epsilon$ for some $y$. Since $A$ is $\epsilon$-sgreedy then $\N{A}(y,t)=0$. By our assumption, this means that $A(y,t)=0$. 
This implies that $v_A(y,t) = 0$, in contradiction.

The claim then follows by $\delta$-redundancy, as $U_A(A,t)$ is a convex combination of such $v_A(x,t)$'s where the coefficients sum to at least $1-\delta$.
\end{proof}

The conclusion in Claim \ref{clm:nothisand that} holds for the first round without any condition. Indeed,  the following claim 
corresponds to Item 2 in the proof of Theorem \ref{thm:greeeeed}.
\begin{claim}\label{clm:firstRound}
$U_A(A,1) \geq (1-\delta)(\maxv(1) - \epsilon)$.
\end{claim}
\begin{proof}
If the condition of Claim \ref{clm:nothisand that} does not hold, it means that there is an $x$ where $A(x,1) = 1$, and so, since $A$ is $\epsilon$-sgreedy, $v_A(x,1) \geq v_A(y,1) - \epsilon$ for all other $y$'s. Therefore, in this case $U_A(A,1) = v_A(x,1) \geq \maxv(1) - \epsilon \geq (1-\delta)(\maxv(1) - \epsilon)$. 
\end{proof}
For the other rounds, 
we cannot prove the same, but similarly to Item 3 in the proof of Theorem \ref{thm:greeeeed},
looking at previous rounds, it is true that:
\begin{claim}\label{clm:decreasing}
For any $s<t$, $U_A(A,s) \geq (1-\delta)(\maxv(t) - \epsilon)$.
\end{claim}
\begin{proof}
We will show that for any $s<t$, and every $x$, $U_A(A,s) \geq (1-\delta)(v_A(x,t) - \epsilon)$.
Take some $x$. If $v_A(x,t) = 0$, then the claim  is clear. Otherwise, $A(x,t) > 0$ or $\N{A}(x,t) > 0$ or both. Either way, $\N{A}(x,s) > 0$. Therefore, as $A$ is $\epsilon$-sgreedy, for every $y$ such that $A(y,s)>0$, 
\[
v_A(y,s) \geq v_A(x,s) - \epsilon \geq v_A(x,t) - \epsilon,
\]  where the last inequality is because $v_A(x,\cdot)$ is non-increasing (Lemma \ref{lem:mono}). 
As $A$ is $\delta$-redundant, $U_A(A,s)$ is a sum of such $v_A(y,s)$'s with coefficients that sum to at least $1-\delta$. 
Hence $U_A(A,s) \geq (1 - \delta)(v_A(x,t) - \epsilon)$, as required.
\end{proof}
Lastly, we 
 prove a claim that
corresponds to Item 4 in the proof of Theorem \ref{thm:greeeeed}.
It implies 
that the utility of a filled box decreases considerably from round 1 to the round when it is filled.
\begin{claim}\label{clm:before}
Assume that $A$ fills box $x$ in time $t>1$. 
Then $v_A(x,t) \leq C(k) v_A(x,1)$.
\end{claim}
\begin{proof}
First,
\[
v_A(x,1) \geq f(x) C(1) \N{A}(x,1)^{k-1} =
f(x) \N{A}(x,1)^{k-1}
\]
On the other hand, since $\N{A}(x,t) = 0$,
\[
v_A(x,t) = f(x) C(k) A(x,t)^{k-1} \leq
f(x) C(k) \N{A}(x,1)^{k-1}
,\]
because $A(x,t) \leq \N{A}(x,t-1) \leq \N{A}(x,1)$.
Combining the above two inequalities gives the result.
\end{proof}
We proceed very similarly to the proof of Theorem \ref{thm:greeeeed}.
Denote by $X_\text{fill}$ the set of rounds for which there is at least one box $x$ that is filled by $A$. Let $X_\text{no-fill}$ be the rest of the rounds, except for $t=1$ which is in neither.
\[
U_A(B) 
\leq \sum_t \maxv(t)
=
\sum_{t \in X_\text{no-fill} \cup \set{1}} \maxv(t)
+
\sum_{t \in X_\text{fill}} \maxv(t).
\]
We want to show that this is not much larger than $U_A(A) = \sum_t U_A(A, t)$.
By Claims \ref{clm:nothisand that}, \ref{clm:firstRound} above, for each $t \in X_{\text{no-fill}} \cup \set{1}$, 
\[
U_A(A,t) \geq (1-\delta)(\maxv(t) - \epsilon).
\]
Therefore, the first sum satisfies:
\[
\sum_{t \in X_\text{no-fill} \cup \set{1}} \maxv(t)
\leq
\sum_{t \in X_\text{no-fill} \cup \set{1}}\B{ \F{U_A(A,t)}{1-\delta}+ \epsilon} 
.\]
Regarding the second sum, denote $X_\text{fill} = \set{t_0, t_1, \ldots, t_k}$ in order. 
By Claim~\ref{clm:decreasing}, for every $i \geq 1$, 
\[
(1-\delta)(\maxv(t_i) - \epsilon) \leq  U_A(A, t_{i-1})
.\]
Therefore,
\[
\sum_{t \in X_\text{fill}} \maxv(t)
\leq
\maxv(t_0) + \sum_{t \in X_\text{fill} \setminus \set{t_k}} U_A(A,t) 
\leq
\maxv(t_0) + \sum_{t \in X_\text{fill}} \B{ \F{U_A(A,t)}{1-\delta}+ \epsilon} 
.\]
To sum up,
\[
U_A(B) 
\leq 
\sum_t \B{ \F{U_A(A,t)}{1-\delta}+ \epsilon}  + \maxv(t_0)
=
\F{U_A(A)}{1-\delta} + \epsilon T + \maxv(t_0)
.\]
By Claims \ref{clm:before} and \ref{clm:firstRound},
\begin{align*}
\maxv(t_0) 
& = \max_x v_A(x,t_0) 
\leq \max_x C(k) v_A(x,1) 
\\ & \leq 
C(k) \B{\F{U_A(A,1)}{1-\delta} + \epsilon}
\leq
C(k) \B{\F{U_A(A)}{1-\delta}} + \epsilon
.\end{align*}
Therefore,
\[
U_A(B) \leq \F{1+C(k)}{1-\delta}U_A(A) + (T+1)\epsilon
.\qedhere\]
\end{proof}

\subsubsection{Polynomial Construction of Approximate Sgreedy and Non-redundancy Strategy}\label{sec:lemGreedyCreate}
\newcommand{\Val}{\mathtt{val}}

\lemGreedyCreate*

This proof is a constructive version of the proof of Lemma \ref{lem:existence}, and so they bear many similarities. 
We construct a matrix $A$ that is as required, calculating its values round after round. 
Intuitively, for time $t$, we examine a candidate value $w$ and construct the $t$-th column of $A$ which satisfies approximately $v_A(t)=w$, where we recall that $v_A(t)
= \min_x\stset{v_A(x,t)}{A(x,t) > 0} 
$. 

The boundary cases are easy. Specifically, if by setting $A(x,t) = 0$ we get $v_A(x,t) \leq w$ then we keep $A(x,t) = 0$. This is because adding more probability to $x$ would only reduce  the value $v_A(x,t)$, and we cannot have a box in the support that is below $v_A(t)$. Similarly, if by setting $A(x,t)=\N{A}(x,t-1)$ we have $v_A(x,t) \geq w$, then we keep 
$A(x,t)=\N{A}(x,t-1)$, which means that the box is filled. Otherwise, we must find a value for $A(x,t)$ that would give (approximately) $v_A(x,t)=w$.
Since $v_A(x,t)$ is monotone as a function of $A(x,t)$, and using a lemma which bounds its derivative from above and below, then using binary search, we can find values of $A(x,t)$ which give a $v_A(x,t)$ that is close enough to $w$. 

Lastly, the sum of these $A(x,t)$'s over the $x$'s needs to be between $1-\delta$ and $1$ for $A$ to be $\delta$-redundant.
We achieve this by a binary search on the value of $w$.

\begin{proof}
We will say a quantity is polynomial if it is as stated in the lemma.
First, if $T>M$, we set $A(\cdot,t) = 0$ for any $t>M$. This trivially satisfies both $\delta$-redundancy and $\epsilon$-sgreediness for these rounds. 
For $t\leq M$,
we show how to calculate $A(x,t)$'s assuming all of $A(x,s)$'s for $s<t$ are already calculated.

Consider $v_A(x,t)$ as a function of $A(x,t)$, with $\N{A}(x,t-1) = 1 - \sum_{s<t} A(x,s)$ fixed. As we know it is a continuous strictly decreasing function, and is defined between $0$ and $\N{A}(x,t-1)$. Denote by $A_w(x,t)$ the extended inverse. That is, given $w$, it is the $A(x,t)$ such that setting it gives $v_A(x,t) = w$ if such $A(x,t)$ exists. If this is not possible, it means either $w$ is too large, and so $A_w(x,t) = 0$, or too small, and then $A_w(x,t) = \N{A}(x,t-1)$. 

The idea is to do a binary search for a good enough $v_A(t) = \min_x \stset{v_A(x,t)}{A(x,t) > 0}$.  Denote our current guess as $w$. We need some procedure to say whether $w$ is too large, too small, or sufficient as a guess for $v_A(t)$.
For this purpose, we find $A(x,t)$'s that approximate the $A_w(x,t)$'s. That is,

\begin{enumerate}
\item If assigning $A(x,t) = 0$ gives $v_A(x,t) \leq w$, we set $A(x,t) = 0 = A_w(x,t)$. 
\item If assigning $A(x,t) = \N{A}(x,t)$ gives $v_A(x,t) \geq w$, we set $A(x,t) = \N{A}(x,t) = A_w(x,t)$.
\item 
Otherwise, we use binary search to find $A(x,t)$ such $|A(x,t) - A_w(x,t)| < \delta/4M$.
As $v_A(x,t)$ as a function of $A(x,t)$ is strictly monotone and continuous this simply involves running a binary search, comparing $v_A(x,t)$ to $w$, until the size of the interval of $A(x,t)$'s we consider is less than $\delta/4M$. This can be done polynomially. 

In fact, we run the binary search even more, so as to be able to guarantee that $|v_A(x,t) - w| < \epsilon/2$. 
In the terminology of Lemma \ref{lem:bounds} (proven below),
$v_A(x,t) =  f(x) \phi(A(x,t))$, and so, if the end points of our current interval are $A_1(x,t)$ and $A_2(x,t)$, then: 
\[
|f(x) \phi(A_1(x,t)) - f(x) \phi(A_2(x,t))| \leq f(1)4^{k-1}|A_1(x,t)-A_2(x,t)|.
\]
Thus, taking $O(k+\log(1/\epsilon))$ steps of the binary search can make this at most $\epsilon/2$. As $w$ is between the values $v_A(x,t)$'s we get for $A_1(x,t)$ and $A_2(x,t)$, we can conclude our search.
\end{enumerate}

Surely, this gives an $\epsilon$-sgreedy matrix at time $t$ (for now it is not doubly-substochastic, but this will be fixed soon).   
Now, consider the sum of the $A(x,t)$'s we got:\\
\begin{enumerate}
\item If it is between $1-\delta$ and $1$ we are done. 
\item If it is greater than $1$ then we consider $w$ as too small, and continue with the next step of binary search. 
\item If it is smaller than $1-\delta$, then we say $w$ is too large and continue.
\end{enumerate}

 If the process concludes, then the $A$ we get is doubly-substochastic by the fact that always $A(x,t) \in [0, \N{A}(x,t)]$, $\sum_x A(x,t) \leq 1$, and for $t>M$, $A(x,t) = 0$. As this $A$ is both $\epsilon$-sgreedy and $\delta$-redundant, we are done. 

Next, we claim that the process concludes and analyze its time complexity. 
Taking $w=0$ will set each $A(x,t)$ to $\N{A}(x,t-1)$, and taking $w$ to be larger than $f(1)$ will set all of them to be $0$. Therefore $\sum_x A_w(x,t)$ can range between 0 and at least $\sum_x \N{A}(x,t-1) \geq 1$ (as $t\leq M$), and by continuity, there is some $w^\star$ such that this sum is exactly $1-\delta/2$. 

Always, $|\sum_x A(x,t) - \sum_x A_w(x,t)| < \delta/4$. This in particular means that 
for any $w$ such that $\sum_x A_w(x,t) \in [1 - 3\delta/4, 1-\delta/4]$ we are guaranteed to stop. Also, it means that if $w>w^\star$ then it will never be considered small, and if $w<w^\star$ it will never be considered large. Thus our binary search is valid and is guaranteed to stop. 

The time it will take to stop is the time until the interval between the $w$'s it considers guarantees that the difference between $\sum_x A_w(x,t)$'s is smaller than $\delta/4$. Say the current interval is $[w_1, w_2]$.
In the terminology 
of Lemma \ref{lem:bounds},
$p = A_{w_2}(x,t)$,   $p +\epsilon = A_{w_1}(x,t)$,
$w_1 = f(x)\phi(p+\epsilon)$, and $w_2 = f(x)\phi(p)$. Therefore, according the the l.h.s.\ of the lemma, 
\[
A_{w_1}(x,t) - A_{w_2}(x,t) = \epsilon \leq 
\BF{(k-1)(w_2 - w_1)}{f(x)(1 - C(k))}^{1/k-1},
\]
which is at most $\delta/4M$ if
\[
w_2 - w_1 \leq \F{f(x)(1-C(k))}{k-1} \BF{\delta}{4M}^{k-1}.
\]
This can be guaranteed with $O( \log(1/f(M)) + \log(1/(1-C(k))) + k\log(1/\delta) + k\log(M))$ binary search steps, as required. \end{proof}

\subsubsection{Upper and Lower Bounds on the Value}

The following Lemma is actually a generalization of the monotonicity Lemma. The $\phi$ defined within is  $v_A(x,t) / f(x)$, where $q$ is $\N{A}(x,t-1)$, and the lemma shows upper and lower bounds on the resulting change of this value, when $A(x,t)$ is changed from $p$ to $p+\epsilon$. 
\begin{lemma} \label{lem:bounds}
Let
\[
\phi(p) = \sum_{i=0}^{k-1} C(i+1) \binom{k-1}{i} p^i (q-p)^{k-1-i}
,\]
where $q \in (0,1]$ and $p \in [0,q]$. For $\epsilon > 0$, where $p+\epsilon \leq q$,
\[
\F{1 - C(k)}{k-1}\epsilon^{k-1} \leq \phi(p) - \phi(p+\epsilon) \leq 4^{k-1} \epsilon
.\]
\end{lemma}
\begin{proof}
Let us first prove the upper bound. Dropping the $C(i+1)$'s as they are at most $1$, and considering the sum as a sum of $2^{k-1}$ terms:
\begin{align*}
\phi(p) - \phi(p+\epsilon)
& \leq 
2^{k-1} \max_{i=0}^{k-1} \left|
(p+\epsilon)^i (q-p-\epsilon)^{k-1-i} - p^i (q-p)^{k-1-i} 
\right|
\\ & \leq
2^{k-1} \max_{i=0}^{k-1} \B{ 
(p+\epsilon)^i (q-p)^{k-1-i}
-
p^i (q-p-\epsilon)^{k-1-i}
}.
\end{align*}
Denoting $a=p$ and $b=q-p-\epsilon$, the term we are maximizing is:
\[
(a+\epsilon)^i (b+\epsilon)^{k-1-i}
-
a^i b^{k-1-i}
\leq
2^{k-1} \epsilon
,\]
where the last inequality is because opening the left term to its $2^{k-1}$ terms, each one is a multiplication of some powers of $a$, $b$ and $\epsilon$, all of them are at most 1. The only one that does not contain $\epsilon$ is canceled out by the right term.
This establishes the upper bound part. 

Next, we prove the lower bound. Denote
\[
B_i(p) 
:= \sum_{j=0}^i \binom{k-1}{j} p^j (q-p)^{k-1-j} 
= q^{k-1} \sum_{j=0}^i \binom{k-1}{j} \BF{p}{q}^j \B{1 - \F{p}{q}}^{k-1-j}.
\]
Then,
\[
\phi(p) = C(k) B_{k-1}(p) + (C(k-1) - C(k)) B_{k-2}(p) + \cdots + (C(1) - C(2)) B_0(p) 
.\]
All of the $B_i$ are non-increasing in $p$, as it is $q^{k-1}$ times the probability that at most $i$ of $k-1$ Bernoulli random variables each of probability $p/q$ are $1$. Also, all of the $C(i-1) - C(i)$ are non-negative, and at least one of them is strictly positive, and is at least $(C(1) - C(k))/(k-1)$. Therefore,
\begin{equation}\label{eq:Bi}
\phi(p) - \phi(p+\epsilon) 
\geq \F{1 - C(k)}{k-1} (B_j(p) - B_j(p+\epsilon)),
\end{equation}
where $j \leq k-2$.
For any integers $i\leq n$, denote by $X^n_i(p)$ the probability that at least $i$ of $n$ i.i.d.\ Bernoulli random variables of probability $p$ are $1$.  Since $B_j(p)$ can be seen as $q^{k-1}(1 - X^{k-1}_{j+1}(p/q))$,  
 Claim \ref{clm:binomLower} (see below) implies that the r.h.s.\ of Eq~\eqref{eq:Bi}
 is at least:
\[
\F{1 - C(k)}{k-1} q^{k-1} \B{\F{p+\epsilon}{q} - \F{p}{q}}^{k-1}
=
\F{1 - C(k)}{k-1} \epsilon^{k-1},\]
as required. 
\end{proof}

\begin{claim} \label{clm:binomLower}
If $i \geq 1$, and $p$ and $p+\epsilon$ are in $[0,1]$, then
\[
X^n_i(p+\epsilon) - X^n_i(p) \geq \epsilon^n.
\]
\end{claim}
\begin{proof} 
We prove it by a double  induction on $i+n$. One base case if when $i=1$ and the second is when $i=n$. 
First, regarding the former base case:
\begin{align*}
X^n_1(p+\epsilon) - X^n_1(p) 
& =
\B{1 - (1-p-\epsilon)^n} - \B{1 - (1-p)^n}
\\ & =
(1-p)^n - (1-p-\epsilon)^n 
=
(a+\epsilon)^n - a^n,
\end{align*}
where $a=1-p-\epsilon \geq 0$. This last expression is easily seen to be at least $\epsilon^n$.
Now, for the latter base case:
\[
X^n_n(p+\epsilon) - X^n_n(p) 
=
(p+\epsilon)^n - p^n 
\geq
\epsilon^n.
\]
For the induction step, first note that if $n > i$,
\[
X^n_i(p) = p X^{n-1}_{i-1}(p) + (1-p) X^{n-1}_i(p).
\]
Therefore, in this case,
\begin{align*}
X^n_i(p+\epsilon) - X^n_i(p) 
& = 
p \B{X^{n-1}_{i-1}(p+\epsilon) - X^{n-1}_{i-1}(p)} \\ & \quad +
(1-p) \B{X^{n-1}_i(p+\epsilon) - X^{n-1}_i(p)} 
\\ & \quad +
\epsilon \B{X^{n-1}_{i-1}(p+\epsilon) - X^{n-1}_i(p+\epsilon)}
\\ & \geq 
p \epsilon^{n-1} + (1-p) \epsilon^{n-1} + \epsilon \cdot 0 =
\epsilon^{n-1} \geq \epsilon^n,
\end{align*}
 where for the first two terms we used the induction hypothesis, and the last term is non-negative by the definition of $X$.
\end{proof}

\section{The Exclusive Policy}\label{sec:exclusive}

\begin{lemma}
\label{lem:astargreedy}
Under the exclusive policy, $\Astar$ restricted to $T\leq M$ rounds is sgreedy and non-redundant.
\end{lemma}
\begin{proof}
See Section \ref{sec:model} for a description of $\Astar$. To see it is non-redundant:
\[
\sum_x \Astar(x,t) = \sum_x \N{\Astar}(x,t-1) - \sum_x \N{\Astar}(x,t) = (M-t-1) - (M-t) = 1
.\]
To see that it is sgreedy, let us fix $t$. The value of a box $x$ is 
\[
v_{\Astar}(x,t) = f(x) \N\Astar(x,t)^{k-1} = f(x) \min\B{1, \alpha(t)q(x)}^{k-1}
=
\min(f(x), \alpha(t)^{k-1}).
\]
If $\Astar(x,t) > 0$ then $\N\Astar(x,t) < 1$,
which means that $v_{\Astar}(x,t) = \alpha(t)^{k-1}$, and so is constant for all such $x$. 
If $\Astar(x,t) = 0$ then $\N\Astar(x,t) = 1$, and so $\alpha(t) \geq 1/q(x) = f(x)^{1/(k-1)}$. Then $v_{\Astar}(x,t) = f(x) \leq \alpha(t)^{k-1}$, as required from a sgreedy  strategy. 
\end{proof}

Hence, according to Theorem \ref{thm:greeeeed}, 
\thmAstar*
According to Claim \ref{lem:greedyCk}  all symmetric equilibria under the exclusive policy are equivalent, and thus equivalent to $A^\star$. Hence, the optimality of $A^\star$ (w.r.t.~symmetric profiles) implies that both the PoSA and PoSS of the exclusive policy are optimal.
That is, for every $f$, $T$, and policy $C$, 
\[
\PoSA(\Cex,f,T)=\PoSS(\Cex,f,T)\leq \PoSS(C,f,T)
.\]

Our next goal is to establish  the PoA of the exclusive policy. For this purpose, we first prove that the success probability of any  equilibrium is at least as large as that of any symmetric profile. 
Since $\Astar$ is a symmetric equilibrium, its optimality among symmetric profiles follows. Hence, the proof provides as alternative proof to the one in \cite{JACM}.

\subsection{Symmetric Equilibria are the Worst}

\thmNashismore*
\begin{proof}
Let $A$ be a  strategy and $\Pnash$ be a profile at equilibrium with respect to  $\Cex$. 
If $\success(\Pnash)=1$, then the inequality is trivial. According to Lemma \ref{lem:non-redundant-0}, we can therefore assume that all players of $\Pnash$ are non-redundant and that $T \leq M$. 
Denote the probability of visiting $x$ in profile $\prof$ by \[\success(\prof, x) = 1 - \prod_{B \in \prof} \N{B}(x,T).\]

We say that box $x$ is \emph{high} with respect to a profile $\prof$ 
if $\success(\prof, x) > \success(A, x)$,
\emph{low} if $\success(\prof, x) < \success(A, x)$,
and \emph{saturated} if they are equal.
The next claim uses the fact that $A$ is symmetric. 

\begin{claim}
\label{claim:nolow}
If a profile $\prof$ is non-redundant and  contains no high boxes, then all boxes are saturated.
\end{claim}
\begin{proof}
If $x$ is not high, then
\begin{equation}\label{eq:AmGm}
\N{A}(x,T)^k
\leq
\prod_{B\in\prof} \N{B}(x,T)
\leq
\B{\R{k}\sum_{B\in\prof} \N{B}(x,T)}^k,
\end{equation}
which means that
\[
\sum_{B \in\prof} \N{B}(x,T) \geq k \N{A}(x,T).
\]
This is the same as
\[
\sum_{B,t} B(x,t) \leq k \sum_t A(x,t).
\]
As there are no high boxes, summing over all $x$'s:
\begin{equation} \label{eq:lastIneq}
\sum_{B,x,t} B(x,t) \leq k \sum_{x,t} A(x,t)
.\end{equation}
As all players in $\prof$ are non-redundant, 
$\sum_x B(x,t) = 1 \geq \sum_x A(x,t)$
for every $t$ and every player $B$. Hence, Eq.~\eqref{eq:lastIneq} is actually an equality. 
On the other hand, Eq.~\eqref{eq:AmGm} is a strict inequality if box $x$ is low and not saturated. Therefore, if even one box is low, we get that Eq.~\eqref{eq:lastIneq} is strict as well, in contradiction. 
\end{proof}

We proceed to prove a weak greediness property for equilibria. Denote a box $x$ {\em full} for player $B$ if $\sum_t B(x,t) = 1$. 
Also, for readability of what follows, when $\prof$ is clear from the context, we shall denote
$
v_B(x,t) = v_{\prof^{-B}}(x,t) = f(x) \cdot \prod_{A \in \prof\setminus\set{B}} \N{A}(x,t)
.$

\begin{restatable}{lemma}{lemWeakGreedy}
\label{lem:weakGreedy}
Consider a profile $\Pnash$ at equilibrium.
For every $B\in\Pnash$ and $t,x,y$ such that $y$ is not full in $B$,
if $B(x,t) > 0$ then $v_B(x,t) \geq v_B(y,t)$.
\end{restatable}
\begin{proof}
Assume otherwise. Define an alternative matrix $B'$ for player $B$, as 
$B' = B + \epsilon( \SM{y,t} - \SM{x,t})$.
For a sufficiently small $\epsilon>0$, $B'$ is a doubly-substochastic matrix
because $y$ is not full in $B$. Then,
$
U_{\prof_{\mathtt{nash}}^{-B}}(B') - U_{\prof_{\mathtt{nash}}^{-B}}(B) 
= \epsilon(v_B(y,t) - v_B(x,t))>0,
$ 
in contradiction.
\end{proof}

Let us define a process that starts with the profile $\Pnash$ and changes it by a sequence of \emph{alterations}, each shifting some amount of probability between two boxes. Importantly, we make sure that each alteration can only decrease the success probability. Hence, the proof is concluded once we show that the final profile has a success probability that is higher than that of $A$. 

We first describe the alternations.
Each alteration considers the current profile $\prof$, and changes it to $\prof'$. 
It takes some high box $x$, some low box $y$ (both w.r.t.\ $\prof$), and the maximal $t$ such that there is a player $B\in\prof$ with $B(x,t) > 0$. 
It defines $B' = B + \epsilon(\SM{y,t} - \SM{x,t})$, and lets the player that played $B$ play $B'$ instead. 
This $\epsilon$ is taken to be the largest so that $x$ does not become low, $y$ does not become high, and such that $\epsilon \leq B(x,t)$, so that the entries remain non-negative. 
Note that $B'$ is doubly sub-stochastic, because taking care that $y$ remains low, also means that $y$'s row in $B'$ still sums to less than $1$. 

After this alteration, either $x$ is saturated, $y$ is saturated, or $B'(x,t)=0$. Clearly, in a finite number of  alterations a profile $\Pfinal$ is obtained, for which either no box is high or no box is low.
\begin{claim}\label{claim:Final}
$\success(\Pfinal) \geq \success(A)$.
\end{claim}
\begin{proof} By Lemma \ref{claim:nolow}, $\Pfinal$ can only contain high and saturated boxes, that is, for every box $x$, $\success(\Pfinal,x) \geq \success(A, x)$. 
However,
$\success(\prof) = \sum_x f(x) \success(\prof, x)$,
and therefore $\success(\Pfinal) \geq \success(A)$. 
\end{proof}
Lastly, the following claim concludes the proof of Theorem \ref{thm:nashismore}.
\begin{claim}\label{claim:alteration}
An alteration can only decrease the probability of success. \qedhere
\end{claim}
\begin{proof}
We first make the following sequence of claims:
\begin{enumerate}
\item \label{it:low}
Consistency: If box $x$ is low (high) w.r.t.\ to some intermediate profile then it was low (high) in all profiles preceding it. 
\item \label{it:monotone}
Monotonicity: Alterations can only increase the value of high boxes, and can only decrease the value of low boxes. This is true in the eyes of all players.
\item \label{it:tT}
If an alteration is made to an intermediate profile $\prof$ with $x,y,B,t$ as above (that is, shifting some probability mass from $B(x,t)$ to $B(y,t)$), then w.r.t.\ to $\prof$,  $v_B(x,T) = v_B(x,t)$.
\item \label{it:xy}
If an alteration is made to an intermediate profile $\prof$ with $x,y,B,t$ as above, then w.r.t.\ to $\prof$, $v_B(x,T) \geq v_B(y,T)$
\end{enumerate}
Here are the proofs:
\begin{enumerate}
\item
This one is clear from the way alterations are defined.
\item
Recall $v_B(x,t) = f(x) \prod_{B' \neq B} \N{B'}(x,t)$. If $x$ is high, then alterations only decrease its $B'(x,t)$ and so increase $\N{B'}(x,t)$, thus increasing its value in the eyes of the different players. This works in the same way, the other way around, for low boxes.
\item
By the way $t$ is chosen in an alteration, $B'(x,s) = 0$ for every player $B'$ and every $s > t$. Therefore, for all $B'$, $\N{B'}(x,t) = \N{B'}(x,T)$, and so
$v_B(x,T) = v_B(x,t)$.
\item
We know that $B(x,t) > 0$ in $\prof$. 
Since $x$ is high for the current profile, then by Item 1 it was also high in all preceding profiles. In particular, $B(x,t)$ was never increased, and so $B(x,t) > 0$ in $\Pnash$ as well. 
Since $y$ is low, then by Item \ref{it:low}, it was also low in $\Pnash$, and so not full. 
Therefore, since $\Pnash$ is an equilibrium, we can apply Lemma \ref{lem:weakGreedy}, and get $v_B(x,t) \geq v_B(y,t)$ w.r.t.\ $\Pnash$. By Item \ref{it:monotone}, this is also true in $\prof$.

By Item \ref{it:tT}, $v_B(x,t) = v_B(x,T)$, and as always $v_B(y,T) \leq v_B(y,t)$, we get the result.
\end{enumerate}

 Now, considering an alteration, let us examine the success probability of a profile $\prof$, and express it as a function of the matrix and values of the player $B$ involved in the alteration:
\begin{align*}
\success(\prof) 
& = 
\sum_x f(x) \B{1 - \prod_{B' \in \prof} \N{B'}(x,T)} 
= 
\sum_x f(x) - \sum_x f(x) \prod_{B' \in \prof} \N{B'} (x,T)
\\ & = 
\sum_x f(x) - \sum_x \N{B}(x,T) v_B(x,T).
\end{align*}
The alteration will increase $B(y,t)$ by $\epsilon$ and decrease $B(x,t)$ by $\epsilon$. As changes are made only to player $B$, all the $v_B$'s are not affected. Thus, 
the change in the success probability as a result of such an alteration is:
\[
\epsilon (v_B(y,T) - v_B(x,T))
.\]
By Item \ref{it:xy} this value is non-positive, and so the alteration can only decrease the success probability.
\end{proof}
Concluding the proof of Theorem \ref{thm:nashismore}.
\qed
\end{proof}

\subsection{The PoA of the Exclusive Policy}

Since $\Astar$ is a symmetric equilibrium, we immediately get that for every $f$ and $T$, the PoA is attained by $\Astar$, that is,
\[
\PoA(C_{ex},f,T)=
 {\max_{\prof\in {\mathcal{P}(T)}}}\BF{{\success}(\prof)}{{\success}(\Astar)}
.\]
Since $\Astar$ has the best success probability among symmetric profiles, and since every policy has some symmetric equilibrium, we get Corollary
\ref{cor:POA}.
To make this more concrete, we show that in the worst case, 
\thmPrice*
Note that as $k$ goes to infinity the PoA converges to $e/(e-1)\approx 1.582$. 
\begin{proof}
Denote $m=\min\{kT,M\}$.
Let $\Aunif$ be the very simple  strategy that chooses at each step uniformly among all the boxes in $\set{1, \ldots, m}$ that it did not choose yet.
The probability that a player checked a specific box $x$ in $\set{1, \ldots, m}$ is $T/m$. 
Therefore,
\[
\success(\Aunif) = \sum_{x=1}^m f(x) \B{1 - \B{1 - \F{T}{m}}^k}
\geq \sum_{x=1}^m f(x) \B{1 - \B{1 - \F{1}{k}}^k}
.\]
Now, by the optimality of $\Astar$,
$\success(\Astar) \geq \success(\Aunif)$, and so:
\[
\PoA(\Cex,f,T) = \F{\sum_{x=1}^m f(x)}{\success(\Astar)} \leq 
\B{1 - \B{1 - \F{1}{k}}^k}^{-1},
\]
where the first equality is by the fact that the PoA of the exclusive policy is attained by $\Astar$. 

Next, consider the situation where $f$ is the uniform distribution on $\set{1, \ldots, M}$, where $M = kT$.
In this specific case, $\Aunif$ is trivially sgreedy, as the value of all boxes at each point in time is the same. Therefore, as it is also non-redundant, by Theorem \ref{thm:greeeeed}, it is a symmetric equilibrium and in particular, an equilibrium. Hence,
\[
\PoA(\Cex,f,T) \geq \F{\sum_{x=1}^M f(x)}{\success(\Aunif)} = \B{1 - \B{1 - \F{T}{M}}^k}^{-1} = 
\B{1 - \B{1 - \R{k}}^k}^{-1}
. \qedhere\] 
\end{proof}

\subsection{Robustness}

The game induced by the exclusive policy is not monotonously scalable, and so  Theorem \ref{thm:robustgeneral}is not applicable to it.  
Indeed, an example follows where the $\Astar$ equilibrium is not robust.
\subsubsection{Non-Robustness of \texorpdfstring{\boldmath{$\Astar$}}{Astar}}
\label{apx:AstarNoRobust}
\begin{lemma}
For every $\epsilon > 0$ there is a configuration where $\Astar_k$ is not an $(1+\epsilon)$-equilibrium when played by $k+1$ players.
\end{lemma}
\begin{proof}
Set $T=1$, and let all boxes except the first have the same $f$.
Denote by $p(x)$ the probability that $\Astar_k$ plays $x$.  Clearly, $p$ is the same for all boxes except the first. Also, by how $\Astar$ is defined $W = M$, and so $p(x) \geq 0$ for all boxes. Denote by $v(1)$ the value boxes get when $\Astar_k$ is played by $k$ players, as we know this is equal for all boxes.
Therefore, when played by $k+1$ players, for all~$x$, 
\[
v(x,1) = f(x)(1-p(x))^{k}=v(1)(1 - p(x)).
\] Denote by $B$ the algorithm that plays box $2$ with probability $1$. Using the equation above we get:
\begin{align*}
\F{U_{\Astar}(B)}{U_{\Astar}(\Astar)}
& =
\F{1-p(2)}{p(1)(1 - p(1)) + (1-p(1))(1-p(2))}
=
\F{1- p(2)}{1-p(1)}\cdot \F{1}{1 + p(1) - p(2)}
\\ & \geq 
\F{1-p(2)}{1-p(1)} \cdot \R{2} 
 = 
\R{2} \cdot \F{\alpha(1)q(2)}{\alpha(1)q(1)}
=
\R{2}\BF{f(1)}{f(2)}^\R{k-1}
.\end{align*}
Therefore, if for example $f(1) = 1/2$, we can set $f(x) = f(2)$ to be as small as we want (by increasing the number of boxes $M$), and this ratio to be as large as we wish.
\end{proof}

\subsubsection{Robustness of \texorpdfstring{\boldmath{$\Astar$}}{A*} under Mild Conditions} 

\thmAstarRobust*
We first prove a simple lemma:
\begin{lemma} \label{lem:simpleLemma}
For non-negative $a_1, \ldots, a_n$ and strictly positive $b_1, \ldots, b_n$, 
$
\sum_i a_i/ \sum_i b_i \leq \max_i a_i/b_i
$.
\end{lemma}
\begin{proof}
First the case $n=2$. Assume by contradiction that $(a_1+a_1)/(b_1 + b_2$) is greater than both $a_1/b_1$ and $a_2/b_2$. 
\[
\F{a_1+a_2}{b_1+b_2} > \F{a_1}{b_1} 
\RIGHT
a_1 b_1 + a_2 b_1 > a_1 b_1 + a_1 b_2
\RIGHT
a_2 b_1 > a_1 b_2
.\]
Also:
\[
\F{a_1+a_2}{b_1+b_2} > \F{a_2}{b_2} 
\RIGHT
a_1 b_2 + a_2 b_2 > a_2 b_1 + a_2 b_2
\RIGHT
a_1 b_2 > a_2 b_1
,\]
in contradiction.
By induction:
\[
\F{\sum_i a_i}{\sum_i b_i} 
\leq
\max \set{\F{a_1}{b_1}, \F{\sum_{i\geq 2} a_i}{\sum_{i\geq 2} b_i}}
\leq
\max_i \F{a_i}{b_i}
. \qedhere\]
\end{proof}
We can now proceed to prove Theorem \ref{thm:AstarRobust}:
\begin{proof}
Recall the definition of $\Astar_k$'s matrix (we will henceforth drop the subscript $k$) given in Section~\ref{sec:model}.
Let $q(x) = f(x)^{-1/
(k-1)}$. For each $t$, 
$\N{\Astar}(x,t) = \min(1, \alpha(t) q(x))
,$
where $\alpha(t) \geq 0$ is such that $\sum_x 1 - \N{\Astar}(x,t) = t$. 

As the $q(x)$ are non-decreasing, then for every $t$ there is some $W_t \leq M$, such that $\N{\Astar}(x,t) < 1$ for every $x\leq W_t$, and $\N{\Astar}(x,t) = 1$ for larger $x$. 
If $\alpha(t) \in [1/q(W_t+1), 1/q(W_t))$, then
\[
t = \sum_x 1 - \N{\Astar}(x,t)
= \sum_{x \leq W_t} 1 - q(x) \alpha(t)
,\]
and so $W_t$ is the largest such that 
\[
\sum_{x\leq W_t} 1 - q(x)/q(W_t) < t
.\]
Note this characterization works also for the case $W_t = M$.
In particular, if $W_t < M$ then 
\[
\alpha(t) \geq 1/q(W_t + 1),
\] 
which we will need later on.
When played with $k+k'$ players, the value w.r.t.\ $\Astar$ is 
\[
v_{\Astar}(x,t) = f(x) \N{\Astar}(x,t)^{k + k' -1}
.\]
Running with $k$ players, for $x \leq W_t$ this is equal to
$v(t) \N{\Astar}(x,t)^{k'}$, where $v(t)$ is the value of all boxes in $\set{1, \ldots, W_t}$ when running $\Astar$ with $k$ players, in which case we know they all have the same value. That is because $A^\star$ is sgreedy, and all of these boxes have $\Astar(x,t) > 0$ (because if $\N{\Astar}(x,t) < 1$, since $\alpha(t)$ is strictly decreasing, we get $\Astar(x,t) > 0$). We also know that for all other boxes, which have $\Astar(x,t) = 0$, by the fact that $\Astar$ is sgreedy when running with $k$ players, $v_{\Astar}(x,t) = f(x) \leq v(t)$.

Consider now an alternative strategy $B$ played by one of the $k+k'$ players. The relative utility this player gains is:
\begin{align*}
\F{U_{\Astar}(B)}{U_{\Astar}(\Astar)}  
& \leq
\F{\sum_t \max_x v_{\Astar}(x,t)}{\sum_t \sum_x \Astar(x,t) v_{\Astar}(x,t)}
\\ & \leq
\F{\sum_t \max_x v_{\Astar}(x,t)}{\sum_t \min_{x\leq W_t}v_{\Astar}(x,t)}
=
\max_t \set{
\F{\max_x v_{\Astar}(x,t)}{
\min_{x\leq W_t} v_{\Astar}(x,t)}}
,\end{align*}
where the last inequality is by Lemma \ref{lem:simpleLemma}.

Fix $t$.
If $W_t = M$, then as mentioned, for all $x$, $v_\Astar(x,t) = v(t)\N{\Astar}(x,t)^{k'} = v(t) (\alpha(t) q(x))^{k'}$. As the $q(x)$ are non-decreasing in $x$,
\[
\F{\max_x v_{\Astar}(x,t)}{
\min_x v_{\Astar}(x,t)}
=
\BF{\alpha(t)q(M)}{\alpha(t)q(1)}^{k'} = \BF{f(1)}{f(M)}^\F{k'}{k-1}
\leq 1+\epsilon
,\]
as desired.
If $W_t < M$, then as mentioned before, $\alpha(t) \geq 1/q(W_t + 1)$. Also, as said, $v_{\Astar}(x,t) = v(t) \N{\Astar}(x,t)^{k'}$ for $x\leq W_t$, and is at most $v(t)$ for $x> W_t$. 
Therefore:
\begin{align*}
\F{\max_x v_{\Astar}(x,t)}{\min_{x\leq W_t} v_{\Astar}(x,t)}
& \leq 
\F{v(t)}{v(t) \N{\Astar}(1,t)^{k'}}
=
\R{(\alpha(t) q(1))^{k'}}
\\ & \leq 
\BF{q(W_t + 1)}{q(1)}^{k'} 
= 
\BF{f(1)}{f(W_t+1)}^\F{k'}{k-1} \leq 1 + \epsilon
,
\end{align*}
 as desired.
\end{proof}

\section{Future Work and Open Questions}

In \cite{JACM}, the main complexity measure was actually the running time and not the success probability. Our results about equilibria are also relevant to this measure, but the social gain is different. For example, it is still true that $\Astar$ is an equilibrium under the exclusive policy, and that all other symmetric equilibria in the exclusive policy are equivalent to it. As $\Astar$ is optimal among symmetric profiles w.r.t.\ the running time, the PoSA of $\Cex$ is equal to the PoSS, and it is also the best among all policies. Furthermore, importing from \cite{JACM}, we know that the PoSA (w.r.t.\ the running time) is about 4.
However, showing the analogue of Corollary \ref{cor:POA}, namely, that the PoA of $\Cex$ is that achieved by $\Astar$, seems difficult, especially because general equilibria are not necessarily greedy.  

Moreover, the results of Vetta \cite{vetta2002nash} do not apply when analyzing the running time, and finding the PoA, PoSA, and PoSS of the sharing policy, for example, remains open.
More generally, it would be interesting to rate policies according to how close they are to the exclusive policy, and analyze their resulting competitive measures as a function of this distance.

Another interesting variant would be to consider feedback during the search. For example, assuming that a player visiting a box $x$ knows whether or not other players were there before. Such a feedback can help in the case that the players collaborate \cite{Dobrev}, but seems to significantly complicate the analysis in the game theoretic variant.

Finally, we would like to encourage game theoretical studies of other frameworks of collaborative search, e.g.,  \cite{Adrian,Yuval-ICALP,ANTS,Higashikawa}.

\clearpage


\bibliographystyle{ieeetr}
\bibliography{bib}

\begin{thebibliography}{10}

\bibitem{JACM}
P.~Fraigniaud, A.~Korman, and Y.~Rodeh, ``Parallel bayesian search with no
  coordination,'' {\em J. ACM}, 2019.

\bibitem{Social-foraging}
L.-A. Giraldeau and T.~Caraco, ``Social foraging theory,'' {\em Princeton
  University Press}, 2000.

\bibitem{hills}
T.~T. Hills, P.~M. Todd, D.~Lazer, A.~D. Redish, I.~D. Couzin, C.~S.~R. Group,
  {\em et~al.}, ``Exploration versus exploitation in space, mind, and
  society,'' {\em Trends in cognitive sciences}, vol.~19, no.~1, pp.~46--54,
  2015.

\bibitem{papanastasiou2017crowdsourcing}
Y.~Papanastasiou, K.~Bimpikis, and N.~Savva, ``Crowdsourcing exploration,''
  {\em Management Science}, 2017.

\bibitem{boinc}
U.~of~California~Berkeley, ``Boinc.'' \url{https://boinc.berkeley.edu/}, 2017.

\bibitem{kleinberg2011mechanisms}
J.~Kleinberg and S.~Oren, ``Mechanisms for (mis) allocating scientific
  credit,'' in {\em Proceedings of the forty-third annual ACM symposium on
  Theory of computing}, pp.~529--538, ACM, 2011.

\bibitem{IFD-review2}
M.~Kennedy and R.~D. Gray, ``Can ecological theory predict the distribution of
  foraging animals? a critical analysis of experiments on the ideal free
  distribution,'' {\em Oikos}, vol.~68, no.~1, pp.~158--166, 1993.

\bibitem{GameTheory}
N.~Nisan, T.~Roughgarden, E.~Tardos, and V.~V. Vazirani, {\em Algorithmic Game
  Theory}.
\newblock Cambridge University Press, 2007.

\bibitem{papa}
E.~Koutsoupias and C.~Papadimitriou, ``Worst-case equilibria,'' {\em Computer
  Science Review}, vol.~3, no.~2, pp.~65 -- 69, 2009.

\bibitem{Welfare}
A.~C. Pigou, ``The economics of welfare,'' {\em Library of Economics and
  Liberty}, 1932.

\bibitem{horn}
R.~A. Horn, {\em Topics in Matrix Analysis}.
\newblock New York, NY, USA: Cambridge University Press, 1986.

\bibitem{birkhoff1946three}
G.~Birkhoff, ``Three observations on linear algebra,'' {\em Univ. Nac. Tacuman,
  Rev. Ser. A}, vol.~5, pp.~147--151, 1946.

\bibitem{nash}
J.~Nash, ``Non-cooperative games,'' {\em Annals of Mathematics}, vol.~54,
  no.~2, pp.~286--295, 1951.

\bibitem{vetta2002nash}
A.~Vetta, ``Nash equilibria in competitive societies, with applications to
  facility location, traffic routing and auctions,'' in {\em Foundations of
  Computer Science, 2002. Proceedings. The 43rd Annual IEEE Symposium on},
  pp.~416--425, IEEE, 2002.

\bibitem{Halpern1}
J.~Y. Halpern, ``A computer scientist looks at game theory,'' {\em Games and
  Economic Behavior}, vol.~45, no.~1, pp.~114--131, 2003.

\bibitem{Halpern2}
J.~Y. Halpern, ``Computer science and game theory: {A} brief survey,'' {\em
  CoRR}, vol.~abs/cs/0703148, 2007.

\bibitem{OmerGame}
R.~Gradwohl and O.~Reingold, ``Fault tolerance in large games,'' {\em Games and
  Economic Behavior}, vol.~86, pp.~438--457, 2014.

\bibitem{collet}
S.~Collet and A.~Korman, ``Intense competition can drive selfish explorers to
  optimize coverage,'' in {\em Proceedings of the 30th on Symposium on
  Parallelism in Algorithms and Architectures, {SPAA} 2018, Vienna, Austria,
  July 16-18, 2018} (C.~Scheideler and J.~T. Fineman, eds.), pp.~183--192,
  {ACM}, 2018.

\bibitem{Gairing}
M.~Gairing, ``Covering games: Approximation through non-cooperation,'' in {\em
  Internet and Network Economics, 5th International Workshop, {WINE} 2009,
  Rome, Italy, December 14-18, 2009. Proceedings}, pp.~184--195, 2009.

\bibitem{Ramaswamy}
V.~Ramaswamy, D.~Paccagnan, and J.~R. Marden, ``The impact of local information
  on the performance of multiagent systems,'' {\em CoRR}, vol.~abs/1710.01409,
  2017.

\bibitem{gittins2011multi}
J.~Gittins, K.~Glazebrook, and R.~Weber, {\em Multi-armed bandit allocation
  indices}.
\newblock John Wiley \& Sons, 2011.

\bibitem{slivkins2019introduction}
A.~Slivkins {\em et~al.}, ``Introduction to multi-armed bandits,'' {\em
  Foundations and Trends{\textregistered} in Machine Learning}, vol.~12,
  no.~1-2, pp.~1--286, 2019.

\bibitem{kremer2014implementing}
I.~Kremer, Y.~Mansour, and M.~Perry, ``Implementing the ``wisdom of the
  crowd'','' {\em Journal of Political Economy}, vol.~122, no.~5,
  pp.~988--1012, 2014.

\bibitem{mansour2015bayesian}
Y.~Mansour, A.~Slivkins, and V.~Syrgkanis, ``Bayesian incentive-compatible
  bandit exploration,'' in {\em Proceedings of the Sixteenth ACM Conference on
  Economics and Computation}, pp.~565--582, 2015.

\bibitem{mansour2016bayesian}
Y.~Mansour, A.~Slivkins, V.~Syrgkanis, and Z.~S. Wu, ``Bayesian exploration:
  Incentivizing exploration in bayesian games,'' {\em arXiv preprint
  arXiv:1602.07570}, 2016.

\bibitem{chen2018incentivizing}
B.~Chen, P.~Frazier, and D.~Kempe, ``Incentivizing exploration by heterogeneous
  users,'' in {\em Conference On Learning Theory}, pp.~798--818, 2018.

\bibitem{crowdsourcing}
Y.~Papanastasiou, K.~Bimpikis, and N.~Savva, ``Crowdsourcing exploration,''
  {\em Management Science}, vol.~64, no.~4, pp.~1727--1746, 2018.

\bibitem{frazier2014incentivizing}
P.~Frazier, D.~Kempe, J.~Kleinberg, and R.~Kleinberg, ``Incentivizing
  exploration,'' in {\em Proceedings of the fifteenth ACM conference on
  Economics and computation}, pp.~5--22, ACM, 2014.

\bibitem{kannan2017fairness}
S.~Kannan, M.~Kearns, J.~Morgenstern, M.~Pai, A.~Roth, R.~Vohra, and Z.~S. Wu,
  ``Fairness incentives for myopic agents,'' in {\em Proceedings of the 2017
  ACM Conference on Economics and Computation}, pp.~369--386, 2017.

\bibitem{monderer1996potential}
D.~Monderer and L.~S. Shapley, ``Potential games,'' {\em Games and economic
  behavior}, vol.~14, no.~1, pp.~124--143, 1996.

\bibitem{rosenthal1973congestion}
R.~W. Rosenthal, ``A class of games possessing pure-strategy nash equilibria,''
  {\em International Journal of Game Theory}, vol.~2, no.~1, pp.~65--67, 1973.

\bibitem{makespan1}
S.~Albers and M.~Hellwig, ``Online makespan minimization with parallel
  schedules,'' {\em CoRR}, vol.~abs/1304.5625, 2013.

\bibitem{makespan2}
A.~Czumaj and B.~V\"{o}cking, ``Tight bounds for worst-case equilibria,'' {\em
  ACM Trans. Algorithms}, vol.~3, pp.~4:1--4:17, Feb. 2007.

\bibitem{makespan4}
S.~Bhattacharya, S.~Im, J.~Kulkarni, and K.~Munagala, ``Coordination mechanisms
  from (almost) all scheduling policies,'' in {\em Innovations in Theoretical
  Computer Science, ITCS'14, Princeton, NJ, USA, January 12-14, 2014},
  pp.~121--134, 2014.

\bibitem{Dobrev}
S.~Dobrev, R.~Kr{\'{a}}lovic, and D.~Pardubsk{\'{a}}, ``Treasure hunt with
  barely communicating agents,'' in {\em 21st International Conference on
  Principles of Distributed Systems, {OPODIS} 2017, Lisbon, Portugal, December
  18-20, 2017}, pp.~14:1--14:16, 2017.

\bibitem{Adrian}
D.~Dereniowski, Y.~Disser, A.~Kosowski, D.~Pajak, and P.~Uznanski, ``Fast
  collaborative graph exploration,'' {\em Inf. Comput.}, vol.~243, pp.~37--49,
  2015.

\bibitem{Yuval-ICALP}
Y.~Emek, T.~Langner, J.~Uitto, and R.~Wattenhofer, ``Solving the {ANTS} problem
  with asynchronous finite state machines,'' in {\em Automata, Languages, and
  Programming - 41st International Colloquium, {ICALP} 2014, Copenhagen,
  Denmark, July 8-11, 2014, Proceedings, Part {II}}, pp.~471--482, 2014.

\bibitem{ANTS}
O.~Feinerman and A.~Korman, ``The {ANTS} problem,'' {\em Distributed
  Computing}, vol.~30, no.~3, pp.~149--168, 2017.

\bibitem{Higashikawa}
Y.~Higashikawa, N.~Katoh, S.~Langerman, and S.-I. Tanigawa, ``Online graph
  exploration algorithms for cycles and trees by multiple searchers,'' {\em J.
  Comb. Optim.}, vol.~28, pp.~480--495, Aug. 2014.

\bibitem{FL69}
S.~D. Fretwell and H.~L.~L. Jr., ``On territorial behavior and other factors
  influencing habitat distribution in birds,'' {\em Acta biotheoretica},
  vol.~19, no.~1, pp.~16--32, 1969.

\bibitem{IFD-review}
T.~Tregenza, ``Building on the ideal free distribution,'' {\em Adv. Ecol. Res},
  pp.~253--307, 1995.

\bibitem{Nicholson}
A.~Nicholson, ``An outline of the dynamics of animal populations.,'' {\em
  Australian Journal of Zoology}, pp.~9--65, 1954.

\bibitem{noga}
N.~Pinter-Wollman, T.~Dayan, D.~Eilam, and N.~Kronfeld-Schor, ``Can aggression
  be the force driving temporal separation between competing common and golden
  spiny mice?,'' {\em Journal of Mammalogy}, vol.~87, no.~1, pp.~48--53, 2006.

\end{thebibliography}


\appendix

\section{Animals Searching for Food}\label{App:animals}

The way animals disperse  in their environment is a cornerstone of ecology  \cite{Social-foraging}.  In these contexts, dispersal is
typically governed by two contradicting forces: The bias towards selecting
the higher quality patches, and the need to avoid costly collisions or overlaps. In the ecology discipline, the (single-round) setting of animals competing over patches of resources has been extensively studied through the concept of {\em Ideal Free Distribution} \cite{FL69,IFD-review2,IFD-review}. Unfortunately, however,  although many animals engage in search over multiple sites, much less is known about relevant game theoretical aspects is such multi-round settings.  

In animal contexts, increased collision costs (the analogy to reward policies) can be caused by various factors, including aggressive behavior, or merely due to equally sharing the patch between the colliding individuals (a.k.a., scramble competition \cite{Nicholson}). Plausibly, collision costs have emerged by evolution due to dynamics that is governed by multiple parameters. Without excluding other factors, one of the possible evolutionary driving forces may relate to competition between groups. Indeed, from the perspective of the group, consuming large quantities of food by all members together (coverage, in our terminology) can indirectly increase the fitness of individuals and hence become significant for their survival. To see why, consider for example, a setting in which two species compete over the same patched food resource, each acting in a different time period of the day (so there is no direct interaction between the two species) \cite{noga}. Assume that  one species is more aggressive towards conspecies. All other factors being similar, at first glance, it may appear that this species would be inferior to the more peaceful one as it induces unnecessary waste of energy and risks of injury. However, perhaps counter-intuitively, our results suggest that it might actually be the converse. Indeed, the higher collision costs of the aggressive species may drive its members to better cover the food resources, on the expense of the more peaceful species.

\section{An Example where a Non-Redundant Sgreedy Strategy is not an Equilibrium}\label{exm:non-eq}

Consider the sharing policy, and take $k=2$, $T=2$, and $f(1) = 4/7$, $f(2) = 3/7$. Denote by $A$ the matrix of the sgreedy non-redundant strategy guaranteed to exist by Lemma \ref{lem:existence}.
By non-redundancy, we know $A = \B{\begin{smallmatrix}
p & 1-p \\ 1-p & p \\
\end{smallmatrix}}
$. Which gives values:
\[
v_A = \B{
\begin{matrix}
\F47 \B{1-p + \F{p}{2}} &
\F47 \cdot \F{1-p}{2} \\
\F37 \B{p + \F{1-p}{2}} & 
\F37 \cdot \F{p}{2}
\end{matrix}}
=
\R{14}\B{
\begin{matrix}
8 \B{1 - \F{p}{2}} &
4(1-p) \\
3 \B{1 + p} & 
3p
\end{matrix}}
\]
Let us put the $1/14$ aside as it only changes utilities by a multiplicative constant.
Setting $p=1$ gives $v_A = 
\B{\begin{smallmatrix}
4 & 0 \\ 6 & 3 \\
\end{smallmatrix}}
$. This means that $A$ is not sgreedy, because in $t=1$, choosing the second box gives a higher utility for this round than $A$ gets.
Setting $p=0$ gives 
$v_A = 
\B{\begin{smallmatrix}
8 & 4 \\ 3 & 0 \\
\end{smallmatrix}}
$. Again, $A$ is not sgreedy, because choosing box 1 in the first round improves the utility at this round.

Otherwise, $0<p<1$, and since $A$ is sgreedy, then $v_A(1,1) = v_A(2,1)$. This means:
\[
8(1-p/2) = 3(1+p)
\RIGHT
5 = 7p
\RIGHT 
p = 5/7
.\]
Therefore, 
$v_A = 
\B{\begin{smallmatrix}
36/7 & 8/7 \\
36/7 & 15/7
\end{smallmatrix}}
$.
Let strategy $B$
first choose box 1 and then box 2. Then, 
\[
\F{U_A(B)}{U_A(A)} = 
\F{36 + 15}{36 + \F27 8 + \F57 15}
>1
.\]
So $A$ is not a symmetric equilibrium.

\section{The PoA of the Sharing Policy}\label{apx:vetta}

The following is practically the same as the proof of \cite{vetta2002nash}, where it is shown that the price of anarchy for valid increasing utility systems is at most $2$. We include it here for completeness, and extend it slightly to hold for approximate equilibria. 
\begin{lemma}\label{lem:vetta}
Consider any policy $C$ such that for all $\ell$, $C(\ell) \leq 1/\ell$. Let $\prof = \set{A_1, \ldots, A_k}$ be some $(1+\epsilon)$-equilibrium. Then, for every   $f$ and $T$,
\[
 {\success}(\prof)\geq \R{2+\epsilon}\cdot \max_{\prof'\in {\mathcal{P}(T)}}{\success}(\prof').
\]
In particular, $\PoA(C,f,T) \leq 2$. 
\end{lemma}

\begin{proof}
Let $\prof = \set{A_1, \ldots, A_k}$ be some $(1+\epsilon)$-equilibrium, and let $\mathbb{O} = \set{O_1, \ldots, O_k}$ be the optimal profile in terms of success probability. Our goal is to show
$\success(\prof) \geq \R{2+\epsilon} \success(\mathbb{O})$.
The trick is to play all of them together.
\begin{align*}
\success(\prof, \mathbb{O})  
& = \success(\prof) + (\success(\prof, O_1) - \success(\prof))
\\ & \quad + \ldots +
(\success(\prof, O_1, \ldots, O_k) - 
\success(\prof, O_1, \ldots, O_{k-1}))
.\end{align*}
We wish to bound the difference $(\success(\prof, O_1, \ldots, O_i) - 
\success(\prof, O_1, \ldots, O_{i-1}))$.
For that, note that when adding a player ($O_i$ in this case) to an existing profile, its utility is always at least what it contributes to the total success probability of the existing profile. This is because for whatever it contributes it gets a utility of $1$, and if it steps on a box that another player already checks, it adds nothing to the success probability, yet may gain some utility (depending on $C$ and the time it checks the box relative to the other players). Therefore,
\begin{align*}
\success(\prof, O_1, \ldots, O_i) & - 
\success(\prof, O_1, \ldots, O_{i-1})
\\ & \leq
U_{\prof,O_1,\ldots,O_{i-1}}(O_i) 
 \leq
U_{\prof^{-A_i}}(O_i)
\leq
(1+\epsilon)U_{\prof^{-A_i}}(A_i)
,
\end{align*}
where the last step is because $\prof$ is a $(1+\epsilon)$-equilibrium.
As $C(\ell) \leq 1/\ell$  for every $\ell$,  $\sum_i U_{\prof^{-A_i}} (A_i) \leq \success(\prof)$, and so:
\[
(2+\epsilon)\success(\prof) \geq 
\success(\prof, \mathbb{O}) \geq \success(\mathbb{O}),\]
as required.
\end{proof}

\end{document}